\documentclass[5p,times]{elsarticle}
\journal{Future Generation Computer Systems}
\usepackage{tikz}
\usepackage{url}
\usetikzlibrary{chains,shapes.multipart}
\usetikzlibrary{shapes}
\usetikzlibrary{automata,positioning}
\usetikzlibrary{calc}
\usepackage{pgfplotstable}
\usepackage{pgfplots}
\usepackage{physics}
\usepackage{amsmath}
\usepackage{dsfont}
\usepackage{tikzscale}
\usepackage{makecell}
\usepackage{tabularx}
\usepackage{tkz-graph}
\usepackage{algorithm, algorithmic}
\usepackage{float}
\usepackage{amsfonts}
\usetikzlibrary{shapes.geometric}
\usetikzlibrary{quantikz}
\usepackage{mathrsfs}
\usepackage{amsthm}
\theoremstyle{definition}
\newtheorem{definition}{Definition}[section]
\newtheorem{theorem}{Theorem}
\theoremstyle{definition}
\usepackage{hyperref}
\hypersetup{
    colorlinks=true,
    linkcolor=cyan,
    filecolor=magenta,      
    urlcolor=blue,
    pdftitle={Overleaf Example},
    pdfpagemode=FullScreen,
    }
\usepackage{draftwatermark}
\SetWatermarkText{Accepted for publication at Future Generation Computer Systems}
\SetWatermarkScale{.15}
\SetWatermarkHorCenter{.97\paperwidth}
\SetWatermarkAngle{90}
\SetWatermarkColor[]{red}
\newlength{\RoundedBoxWidth}
\newsavebox{\GrayRoundedBox}
\newenvironment{GrayBox}[1][\dimexpr\columnwidth-4.5ex]%
   {\setlength{\RoundedBoxWidth}{\dimexpr#1}
    \begin{lrbox}{\GrayRoundedBox}
       \begin{minipage}{\RoundedBoxWidth}}%
   {   \end{minipage}
    \end{lrbox}
    \begin{center}
    \begin{tikzpicture}%
       \draw node[draw=black,fill=gray!10,rounded corners,%
             inner sep=2ex,text width=\RoundedBoxWidth]%
             {\usebox{\GrayRoundedBox}};
    \end{tikzpicture}
    \end{center}}
\usepackage{bm}
\usepackage{subcaption}
\usepackage{comment}
\begin{document}

\begin{frontmatter}
\title{Paving the Way to Hybrid Quantum-Classical Scientific Workflows}

\author[inst1,inst2]{Sandeep Suresh Cranganore}
\affiliation[inst1]{organization={Forschungszentrum Jülich GmbH, Peter Grünberg Institute, Quantum Control (PGI-8)},
            addressline={Wilhelm Johnen Straße}, 
            city={Juelich},
            postcode={52425}, 
            country={Germany}
            }
\affiliation[inst2]{organization={TU Wien, HPC group},
            addressline={Favoritenstrasse 9-11}, 
            city={Vienna},
            postcode={1040}, 
            country={Austria}}

\author[inst2]{Vincenzo De Maio}

\author[inst2]{Ivona Brandic}

\author[inst3]{Ewa Deelman}

\affiliation[inst3]{organization={University of Southern California, Information Sciences Institute},
            addressline={4676 Admiralty Way Suite 1001
}, 
            city={Marina del Rey},
            postcode={90292}, 
            state={California},
            country={USA}}

\begin{abstract}
    The increasing growth of data volume, and the consequent explosion in demand for computational power, are affecting scientific computing, as shown by the rise of extreme data scientific workflows. As the need for computing power increases, quantum computing has been proposed as a way to deliver it. It may provide significant theoretical speedups for many scientific applications (i.e., molecular dynamics, quantum chemistry, combinatorial optimization, and machine learning).  Therefore, integrating quantum computers into the computing continuum constitutes a promising way to speed up scientific computation.  However, the scientific computing community still lacks the necessary tools and expertise to fully harness the power of quantum computers in the execution of complex applications such as scientific workflows. In this work, we describe the main characteristics of quantum computing and its main benefits for scientific applications, then we formalize hybrid quantum-classic workflows, explore how to identify quantum components and map them onto resources. We demonstrate concepts on a real use case and define a software architecture for a hybrid workflow management system. 
\end{abstract}

\end{frontmatter}

\section{Introduction}
Scientific computing is a branch of computer science spanning different disciplines (i.e., biology, chemistry, engineering), whose goal is the development of standardized and accurate simulations of different phenomena. Scientific computation is typically modeled by scientific workflows, whose execution is managed by Workflow Management Systems (WMSs) (e.g., Pegasus~\cite{Deelman2019}, ASKALON~\cite{askalon-ref}, Airflow~\cite{haines2022workflow}). The importance of scientific workflows has been proven by the Nobel-Prize-winning research on gravitational waves, which employed LIGO data analysis workflows managed by Pegasus WMS\footnote{\href{https://today.usc.edu/nobel-prize-winning-discovery-on-gravitational-waves-came-about-with-contributions-from-usc-scientists/}{Nobel Prize-winning discovery on gravitational waves came about with contributions from USC scientists
}}~\citep{PhysRevLett.118.221101}, and their broad applicability to critical fields~\cite{taylor2007workflows} such as drug design~\cite{vanhaelen2017design}, material sciences~\cite{stein2019progress}, and simulations of the spread of Covid-19~\cite{ozik2021population}. The increasing complexity of scientific workflows calls for increasing computing power, which is provided by HPC clusters~\cite{Abbott_2017}. 

Current HPC systems are faced with the end of Moore's law~\cite{shalf2020future}. This means that we cannot increase the computing power at the same rate as before. Consequently, HPC researchers are considering alternative forms of computing to satisfy the increasing demands of scientific applications, enabling the transition to Post-Moore scientific computing~\cite{ashby2010opportunities}. 

In the landscape of Post-Moore computing, quantum computing promises substantial performance improvement~\citep{Ang2022, 10.1145/3458817.3487399}. 
Quantum computing can increase application performance, due to the proven theoretical speedup for different scientific problems~\cite{givi2020quantum,liu2021rigorous} and its native modeling of many scientific phenomena~\cite{Chang2021}. However, despite the theoretical speedup, the current state-of-the-art hardware is bound by the following shortcomings (1) limited availability of hybrid resources, (2) susceptibility of qubits to noise and errors in the \emph{Noisy Intermediate-Scale Quantum} (NISQ) devices~\citep{Aaronson2015-wh, Cheng2023-ra, Preskill2018quantumcomputingin}, (3) limited technical capabilities and engineering shortcomings at the hardware level: several requirements, such as highly-controllable qubits (high-fidelity state preparation and qubit register initialization), large counts of quantum gates (for e.g. deeper quantum circuits have higher CNOT (controlled-NOT) counts, which contribute to larger error rates as compared to single-qubit gates) operating within the coherence limits of the qubits (for e.g. shorter gate times and efficiency), and considerably large circuit depths (starting from qubit initialization to the final measurement). In order to run quantum algorithms (cryptosystems-based algorithms, such as Shor's integer factorization algorithm for breaking the RSA-2048 scheme~\cite{365700}) require large-scale logical qubit devices or alternatively physical qubits and quantum gates ranging between thousands to millions~\cite{Gidney2021howtofactorbit}, and (4) challenges in suppressing errors, for e.g. protecting entanglement between logical qubits~\cite{Cai2024-nb} (fault-tolerant quantum computation) arising due to the lack of fault-tolerant logical algorithms/quantum error correction (QEC) schemes at the experimental level. 

A broad category of workflows tasks can be propelled by utilizing quantum processors. These range from (a) accelerators that are interoperable with classical architechtures, resulting in hybrid quantum-classical systems~\cite{stein2021hybrid} or neuromorphic architectures~\cite{10.1063/5.0020014}, (b) stochastic and probabilistic sampling methods such as Monte Carlo (MC) estimation, (c) Programmable array of qubits, for synthetic simulation of other quantum systems, also known as \emph{quantum simulators}, e.g., quantum phases of matter and critical dynamics of many-body systems~\citep{Keesling2019-br, Ebadi2021}. 

In the field of gate-based NISQ, the most influential paradigm subclass of hybrid models are the \emph{Variational Quantum Algorithms} (VQAs)~\cite{Cerezo2021-xn}, where classic and quantum hardware are tightly copuled and cooperate in the achievement of a specific task. VQAs are turning out to be one of the much anticipated workhorses in the hybrid computation arena. These include fluid dynamics, quantum chemistry simulations, for e.g., accurate calculations of electronic structure using Hartree-Fock methods~\cite{Quantum2020}. Molecular dynamics (MD) is another highly suitable usecase, for e.g., simulating weakly bound, coarse-grained intermolecular interactions and groundstate determination~\cite{PhysRevA.105.062409}. 

In this work, we investigate the problem of executing scientific workflows on hybrid quantum-classical ecosystems. First, we identify and formalize the main actors involved in the process. Based on our model, we design a hybrid quantum-classic workflow starting from a classic molecular dynamics workflow designed for Pegasus WMS. Then, we provide an idea of how to allow the execution of scientific workflows on hybrid quantum-classic systems and identify challenges and possible solutions. Finally, we provide an outlook on the field and identify possible trends for future research in the area. 

We focus on scientific applications, which provide major opportunities for quantum modeling and quantum speedup~\cite{stein2021hybrid}. Also, we consider Pegasus~\cite{Deelman2019}, a well-known WMS, as the reference architecture for WMSs. 

The paper is organized as follows: first, we analyze related work in Section~\ref{sec:related}, then we provide the theoretical foundations of our work in Section~\ref{sec:background} (Appendixes provide additional background in quantum computing). In Section~\ref{sec:classic-hybrid}, we provide the definition of hybrid quantum-classic workflows and how to transform a classic workflow into a hybrid quantum-classic workflows. In Section~\ref{sec:usecase}, we describe our molecular dynamics simulation use case as a running example of the transformation from classic to hybrid classic-quantum workflows. We then describe our vision of hybrid workflow execution in Section~\ref{sec:hybrid-wms}, while in Section~\ref{sec:challenges} we identify the challenges that must be tackled to enable it. Finally, we conclude our paper in Section~\ref{sec:outlook}. 

\section{Related Work}
\label{sec:related}
Quantum computing has been first theorized by Feynman~\cite{Feynman1982-uz}. Advantages and ideas for quantum supremacy over superconducting qubits are described in~\cite{Arute2019}, while~\cite{Leymann2020} describes the circuit-based model of computation. Similarly,~\cite{Kielpinski2002} focuses on ion-traps, while~\cite{Willsch2022} describes D-Wave quantum annealers.

In~\cite{weder2022analysis}, a first study on how to transform classic workflows into hybrid classic/quantum workflows is performed. In particular, it focuses on machine learning applications~\cite{vietz:22} and also describes methods for the identification of quantum candidates and splitting scientific applications between classic and quantum hardware. Applications of quantum computing to scientific applications can be found in many domains, ranging from drug design~\cite{drugquantum}, molecular dynamics~\cite{Cranganore2022}, financial modelling~\cite{orus2019quantum}, manufacturing industry~\cite{manifacturingquantum}, linear optimization~\cite{PhysRevLett.103.150502}, and healthcare~\cite{healthcarequantum}. In the above examples, scientific workflows are not considered.  

VQAs, one of the typical applications of hybrid classic-quantum systems, are described in~\cite{McClean_2016}. In~\cite{cerezo2021variational}, different applications of VQAs are described. Still in the context of variational quantum algorithms,~\cite{Peruzzo2014-uw} focuses on photonic quantum platforms, while~\cite{Tilly2021,Tilly2022-fv} focus on the variational quantum eigensolver, that is a common task in scientific computations. Applications of VQAs can be found in different scientific applications, such as molecular dynamics~\cite{Cranganore2022}, accelerating machine learning workloads~\cite{liu2021rigorous} and combinatorial optimization~\cite{McClean_2016}. However, few works address the integration of VQAs in scientific workflows.  

Efforts in standardizing hybrid applications are performed also from a software engineering perspective. In~\cite{destefano:22}, a survey about the state of the art in quantum software engineering is performed. Works like~\cite{weder:22} focus on the software development cycle for quantum applications. Other approaches~\citep{atkinson2019quantum, Davis2020TowardsOT} focus on automatic synthesis of quantum programs.

Classical workflow management systems such as Pegasus, ASKALON, and DagsHub\footnote{\url{https://dagshub.com/}} are described respectively in~\citep{pegasus-ref,askalon-ref}. Execution of workflows in hyper-heterogeneous architectures is described in~\cite{maio2020a}. Also,~\cite{silva2017a} proposes a characterization of workflow management systems for data-intensive applications. Support for quantum workflow is provided in tools such as Orquestra~\cite{orquestra} and Covalent~\cite{covalent-ref}. In this work, we provide guidelines on how to adapt existing classical scientific workflows into hybrid quantum-classical scientific workflows and how to extend existing classical WMS to integrate both quantum and classical hardware.

We extend the outlined works by focusing on the integration of quantum machines in the execution of HPC applications. We generalize the concept of hybrid quantum-classical workflows, defining different execution models for hybrid quantum-classical applications and validating our findings on a real-world molecular dynamics simulation workflow. Based on our findings, we identify open challenges and possible solutions for the integration of quantum devices in HPC applications. 

\section{Background}
\label{sec:background}

\subsection{Scientific Workflows}

Scientific Applications in different domains (i.e., finance, biology, chemistry, engineering) can be decomposed in elementary \emph{tasks} (i.e., aggregate data from different sources, average a set of samples, apply a method to a specific dataset). Tasks can be combined into \emph{workflows}, represented as  directed acyclic graphs (DAGs)~\cite{pegasus-ref,askalon-ref,dagshub-ref} where nodes represent the tasks and edges model data and control dependencies between tasks. 
\begin{definition}[Scientific Workflows]
    A workflow $W$ can be formally defined as a DAG $W = (T, E)$, such that $T$ is a set of tasks and $E$ the set of edges, with $E \subset T \times T$.\label{definition:classic-wf}
\end{definition}

Workflows modelling scientific applications are called \emph{scientific workflows}. Workflows and tasks can be stored in public repositories (i.e., Pegasus workflow gallery\footnote{\url{https://pegasus.isi.edu/workflow_gallery/}}), allowing re-use of validated code, \emph{repeatability} of simulation, (possibility to easily repeat the setup and execution of a simulation), which increases confidence in simulation's results, and reproducibility of computation (possibility to reproduce and verify results of computation), creating opportunity for new insights and reducing measurements errors. Also, workflows are fundamental for the development of \emph{standardized}, \emph{robust}, and \emph{accurate} simulations of different phenomena.

\subsection{Workflow Management Systems}
Execution of scientific workflows on HPC infrastructures requires different software layers, to enable (1) scheduling of workflow tasks onto different computing resources, (2) management of data, including intermediate data products (either streaming data, or scientific datasets), (3) interoperation between different heterogeneous resources (e.g., Cloud/Edge nodes, academic clusters), and (4) fault tolerance (e.g., checkpointing of execution, re-execution of tasks). 

\subsection{Quantum Computing}
Recent years have seen a major boom in the areas of quantum information processing and quantum technologies. The huge surge in academic interest and industrial investment in quantum happened more or less after the seminal publication by \texttt{Google Inc.} on \emph{Quantum supremacy using a programmable superconducting processor}~\cite{Arute2019}. Along with neuromorphic computing architectures, quantum computation, and quantum simulation have emerged as some of the most promising paradigms in alternative computing architectures. Multiple quantum platforms based on superconducting qubits like the IBMQ universal quantum computer~\cite{Leontica2021}, programmable atomic arrays~\cite{Ebadi2021}, trapped-ion quantum computers~\cite{Kielpinski2002}, \emph{D-Wave} 2000Q and 5000Q quantum annealers~\cite{Willsch2022} exist today with the promise of accelerating a wide range of problems that would typically be impossible to solve or simulate on a classical (hardware and software) computers. Some of these include Shor's groundbreaking integer factorization  algorithm~\cite{365700} offering superpolynomial speedup, Harrow-Hassidim-Lloyd (HHL) algorithm for solving a large system of sparse matrices with exponential speedup~\cite{PhysRevLett.103.150502}, and accelerated linear algebra computations like matrix multiplication~\cite{Zhang2016}. Other highly relevant domains include combinatorial optimization (NP-hard problems), finance, machine learning~\citep{Biamonte2017, RevModPhys.91.045002, Liu2021}, battery design, new novel molecule and drug discovery, quantum materials, grid power management to name a few. 

Although the potential to accelerate time to solution using quantum machines is huge, quantum computation is still in its very beginning, suffering from many hardware and software imperfections. Currently, the technology for initial state preparation of quantum registers, precise qubit control, high-fidelity quantum gate preparation, and measurement of qubits involves a high level of uncertainty. These can be traced mainly to ultra-precise engineering bottlenecks and environmental errors induced by \emph{decoherence}. For example, the measured fidelity of 2 million samples on the 53-qubit \emph{Sycamore} chip fabricated by \texttt{Google} Inc.~\citep{Arute2019, zlokapa2020boundaries}, described in terms of the linear cross-entropy benchmark (XEB), is only at a level of 0.2 \%~\citep{Arute2019, 10.1145/3458817.3487399}. This does not keep up with the performance of classical simulators, which can be exponentially complex, but provide higher fidelity compared to quantum devices. 

\subsection{Hybrid Quantum-Classic Systems}
\begin{figure}[!h]
    \centering
    \includegraphics [width=0.95\columnwidth]{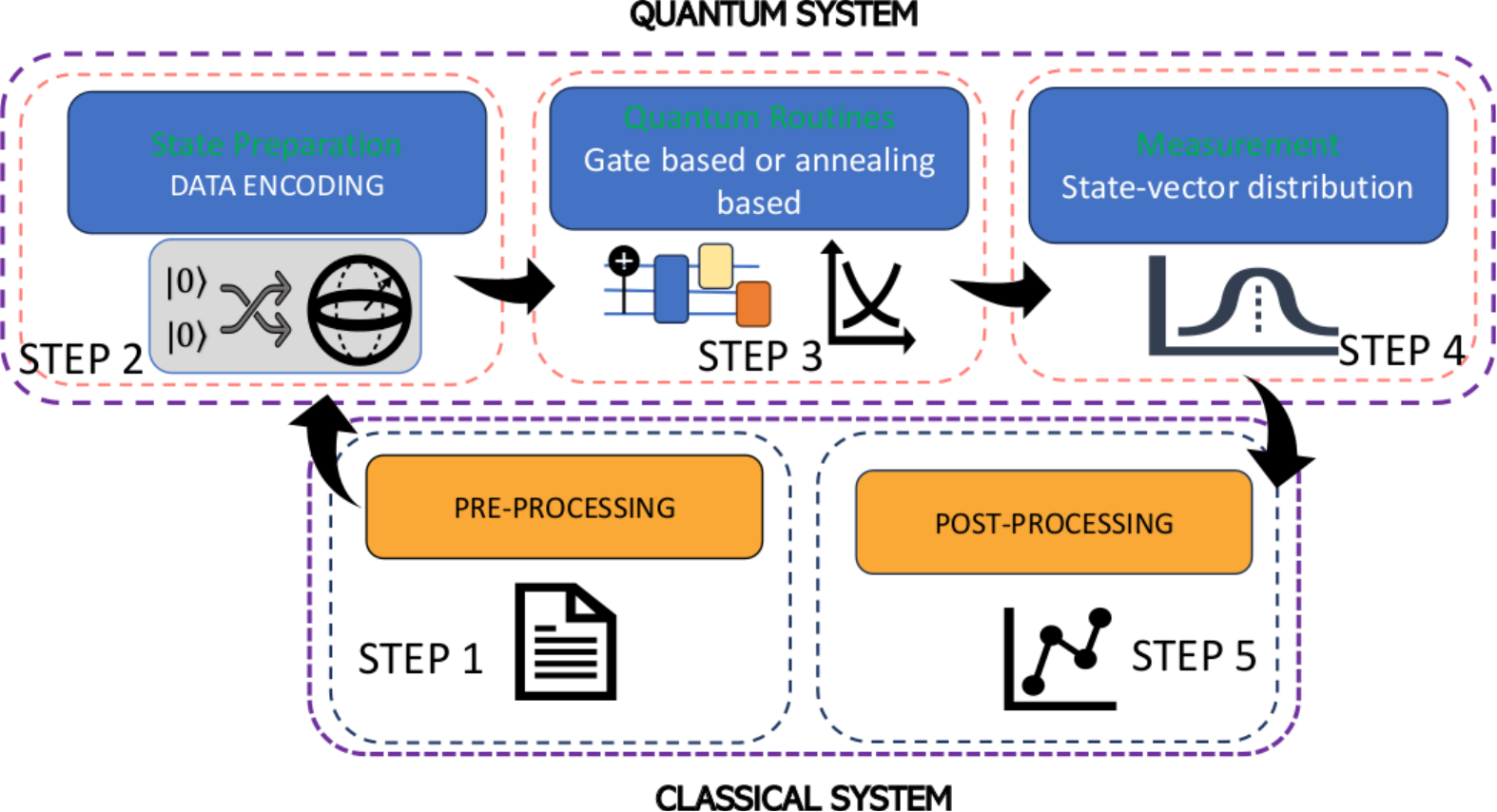}
    \caption{Schematics of hybrid quantum-classical systems.}
    \label{fig:hybrid_systems}
\end{figure}


Hybrid quantum-classic systems define a new class of computing paradigms combining the computing capacities of near-term noisy quantum processors and classical co-processors working in conjunction to solve large-scale scientific problems. The main advantage of using this approach is that it allows the exploitation of the many strengths of classic processors for multiple tasks (e.g., convex optimization, error correction, data pre/post-processing) while at the same time utilizing the capabilities of quantum machines for other specific tasks~\cite{stein2021hybrid}. The motivation behind hybrid ecosystems is to orchestrate quantum algorithms with classical routines which are more suited and efficient on classical processors (for e.g. classical optimizers for back-propagation computation, data entry, graphics, data pre/post-processing etc.). This allows for the distribution of larger workloads to classical devices, thereby mitigating the burden on error-prone quantum hardware, which is leveraged for specialized and targetted tasks (for instance, quantum phase estimation (QPE), cryptographic schemes, operator quantum expectation value computation, optimization, etc.). This strategic integration yields a substantial reduction in the utilization time of the quantum resources, enhancing the efficacy of task executions. Figure~\ref{fig:hybrid_systems} describes the hybrid system pipeline: In step 1, data is pre-processed on the classic system for further encoding onto the quantum registers; in Step 2, the quantum state is prepared based on preprocessed input (typically done using data encoding schemes), Step 3, manipulates a quantum circuit; in step 4, the quantum state is measured and post-processed in step 5.

\textbf{Data Encoding For Quantum Devices}: Classical data in its raw form cannot be processed on quantum devices. Desired initial quantum \emph{state preparation} necessitates converting and encoding the input classical vector (tensor) data in a suitable representation as \emph{quantum data} for embedding, storing the quantum information in the QPUs, and performing quantum operations via quantum algorithms. Porting classical data sets onto quantum devices can be achieved efficiently using multiple \texttt{DATA ENCODING} schemes~\citep{weigold2021, PhysRevA.102.032420}. In general, choosing a particular data encoding method depends on the use case (model/algorithm dependent). Some of the well-known encoding frameworks are (i) \texttt{Basis} Encoding, (ii) \texttt{qRAM} Encoding, (iii) \texttt{Angle} Encoding, (iv) \texttt{Amplitude} Encoding, etc. 

\textbf{Parametrized Quantum Circuits}: Parametrized quantum circuits (PQCs) or variational circuits are basically quantum algorithms that vary certain variables (real/complex valued vectors) or parameters, often denoted as  $\boldsymbol{\vartheta}$. PQCs like any other quantum circuit consists of, (1) Initialized qubit register (cf. Appendix B, section(\ref{quantum-registers}) with appropriate quantum state preparation, (2) a quantum circuit with a cascade of quantum logic gates (cf.  \hyperref[quantum-gates]{Appendix C} and \hyperref[complex-quantum-circuits]{Appendix D}) $U(\boldsymbol{\vartheta})$, parameterized by a set of  parameters $\boldsymbol{\vartheta}$ (3) classical optimizers augmenting the PQCs and, (4) measurements and resets.

\textbf{Variational Quantum Algorithms} (VQAs) are one of the most important paradigms in hybrid quantum-classical systems and are prime candidates for quantum advantage. VQAs can be defined as hybrid quantum-classical algorithms, wherein a parametrized quantum circuit is iteratively optimized via classical optimization algorithms. The schematic diagram Figure~\ref{fig:hybrid_systems} shows that the black box performs VQA executions. The black box can be decomposed into two major blocks, namely the quantum block and a classical block, interconnected by an underlying adaptive feedback-loop mechanism. 

\begin{itemize}
    \item \textbf{Quantum System}: A \emph{Noisy-Intermediate-Scale-Quantum} (NISQ) device that at the low level prepares highly entangled parameterized quantum states and performs 
    \newline 
    quantum-subroutines using PQCs (variational circuits). For e.g. this block executes a \emph{forwardpass} by computing the quantum expectation value of certain physical observables
    \newline 
    /operators (matrices) and the measured quantum state yields the corresponding parameter values which are stored in the memory at every intermediate step.
    
    \item \textbf{Classic System}: A \emph{classical} optimizer which receives the quantum outputted parameters and executes iterative optimization, (for e.g. gradient descent) by optimizing the cost (loss) function landscape. Calling gradients functions for parameter $\boldsymbol{\vartheta}_i$ update is done using \newline  
    \emph{back-propagation} (accumulated gradients). The data flow in this stage corresponds to passing the updated parameters loop back into quantum circuit for further quantum function calls (gate operations) and initiate subsequent control flow steps.
\end{itemize}

\textbf{Other Hybrid Algorithm Frameworks} 
Apart from VQAs, there are other classes of hybrid algorithms, typically analog-based, that leverage existing quantum architectures. These include hybrid forms of quantum annealing, such as hybrid solvers in D-WAVE machines for solving arbitrary structure QUBO problems (quadratic models, such as Binary Quadratic Models (BQMs), Constrained Quadratic Models (CQMs) or Unconstrained Quadratic Models, Discrete Quadratic Models (DQMs)) etc)\footnote{\url{https://docs.dwavesys.com/docs/latest/doc_leap_hybrid.html}}. Frameworks such as \emph{iterated} adiabatic reverse annealing which are quantum annealing-based techniques embedded in a classical loop have also been useful in tackling multiple scientific and industrial use cases~\cite{PhysRevA.106.010101}. Some more ongoing developments in the hybrid algorithms sectors are \emph{Quantum Neuromorphic Computing}, wherein brain-inspired classical neural network architectures are conjoined with quantum hardware to offer computational advantage~\cite{10.1063/5.0020014}. 

\section{From Classic to Hybrid Workflows}
\label{sec:classic-hybrid}
In this section, we provide the main definitions of what is needed to enable our vision of hybrid workflow execution, focusing on hybrid workflows and then on hybrid WMSs.

\subsection{Hybrid Workflows}
\label{sec:hy}
\begin{figure*}[!ht]
\includegraphics[width=\textwidth]{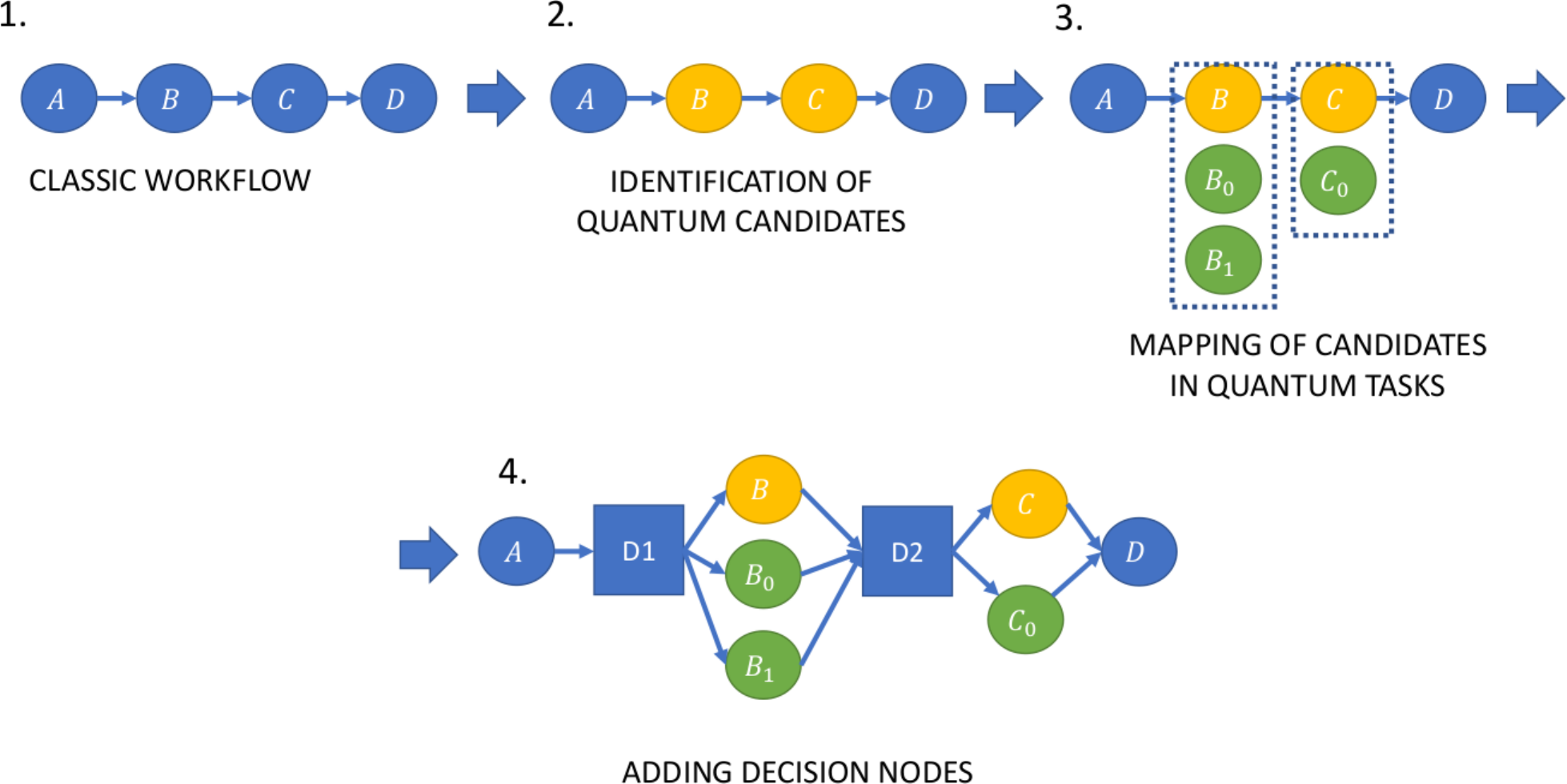}
\caption{Example of Transformation into Hybrid Workflow}
\label{fig:hybrid-wfs}
\end{figure*}
We extend the definition of classic workflows by adding a set of \emph{quantum} tasks $Q$ to Definition~\ref{definition:classic-wf}. Quantum tasks $q \in Q$ that can be executed only by quantum machines. Also, quantum tasks in our workflow definition are \emph{functionally equivalent}~\cite{Shashidhar2005equivalence} to some tasks in $T$. Because of the undecidability of the program equivalence problem~\cite{Shashidhar2005equivalence}, we assume that the user defines a mapping function $f: T \mapsto Q$ that maps a classic task into its quantum equivalent. Multiple quantum tasks can be available for different classic tasks: for example, for the classic task of computing a matrix eigenvalue, either HHL~\cite{zhang2022improved} or VQE~\cite{cerezo2021variational} algorithm can be used, depending on available quantum hardware. As a consequence, $f$ is surjective, but not injective. We define \emph{quantum candidates} tasks as tasks for which there is a quantum task, namely, 
\begin{equation}
\forall t \in T : f(t) \neq \emptyset, t~\text{is a quantum candidate}.
\end{equation}
We define the set $T'$ as the set of quantum candidates. We assume that $T' \subseteq T$. If a task $t'$ is a quantum candidate, we add a \emph{decision} node between $t'$ predecessor and connect it to $t'$ and all other tasks $q \in f(t')$. Different execution paths will be executed according to \emph{conditions} or specified by the decision node. An example of a condition could be, for example, to execute the quantum task if quantum hardware is available. The set of decision nodes is defined as $D$. Figure~\ref{fig:hybrid-wfs} provides an example of how to transform a classic workflow into a hybrid workflow.

\begin{figure*}[!ht]
    \centering
    \begin{subfigure}{.49\textwidth}
    \includegraphics[width=\textwidth]{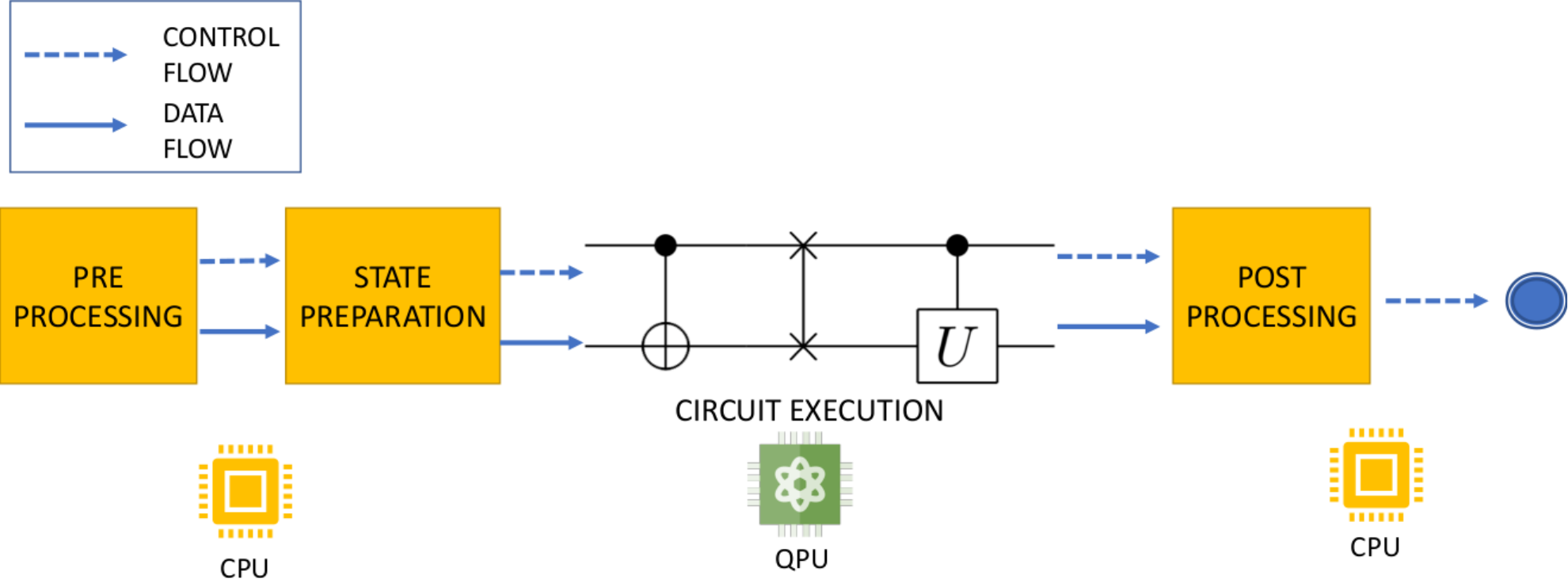}
    \caption{Circuit Execution}
    \label{fig:circuit-execution}
    \end{subfigure}
     \begin{subfigure}{.49\textwidth}
    \includegraphics[width=\textwidth]{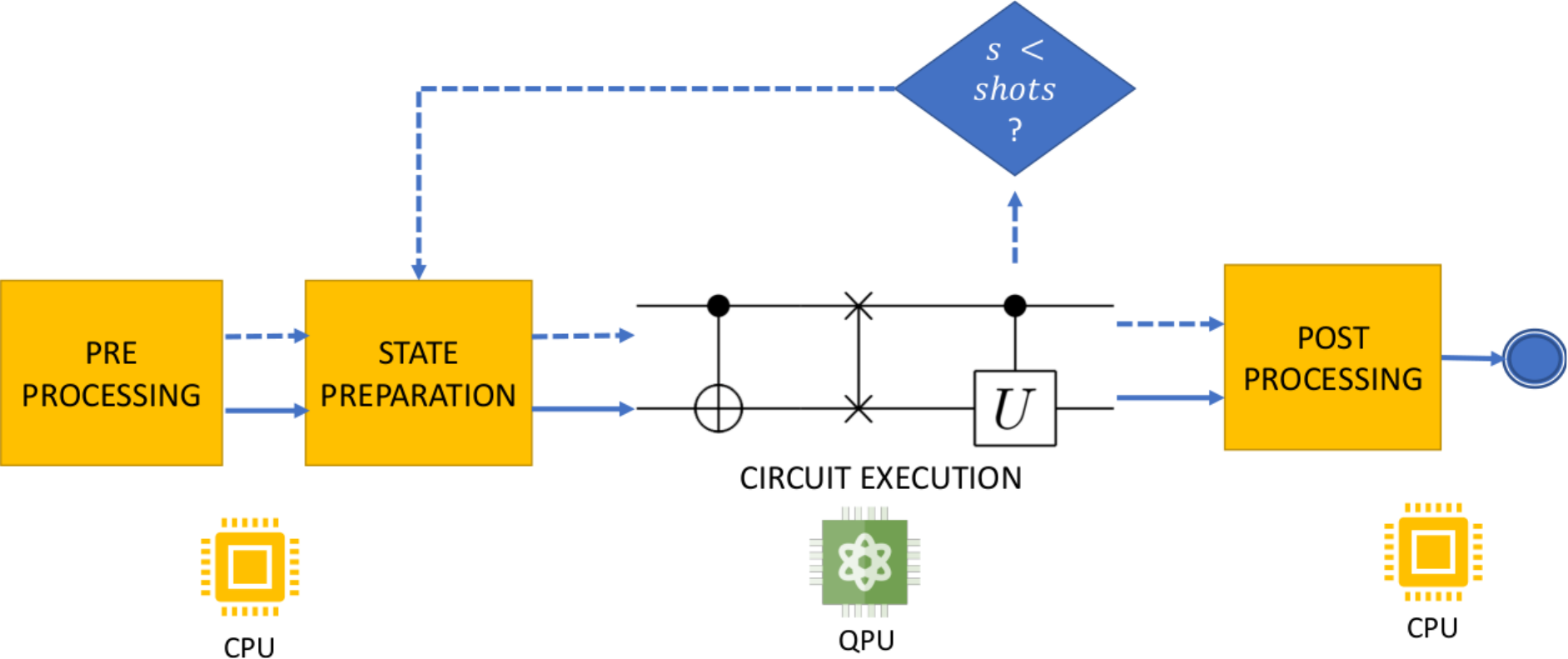}
    \caption{Task Execution}
    \label{fig:task-execution}
    \end{subfigure}
     \begin{subfigure}{.49\textwidth}
    \includegraphics[width=\textwidth]{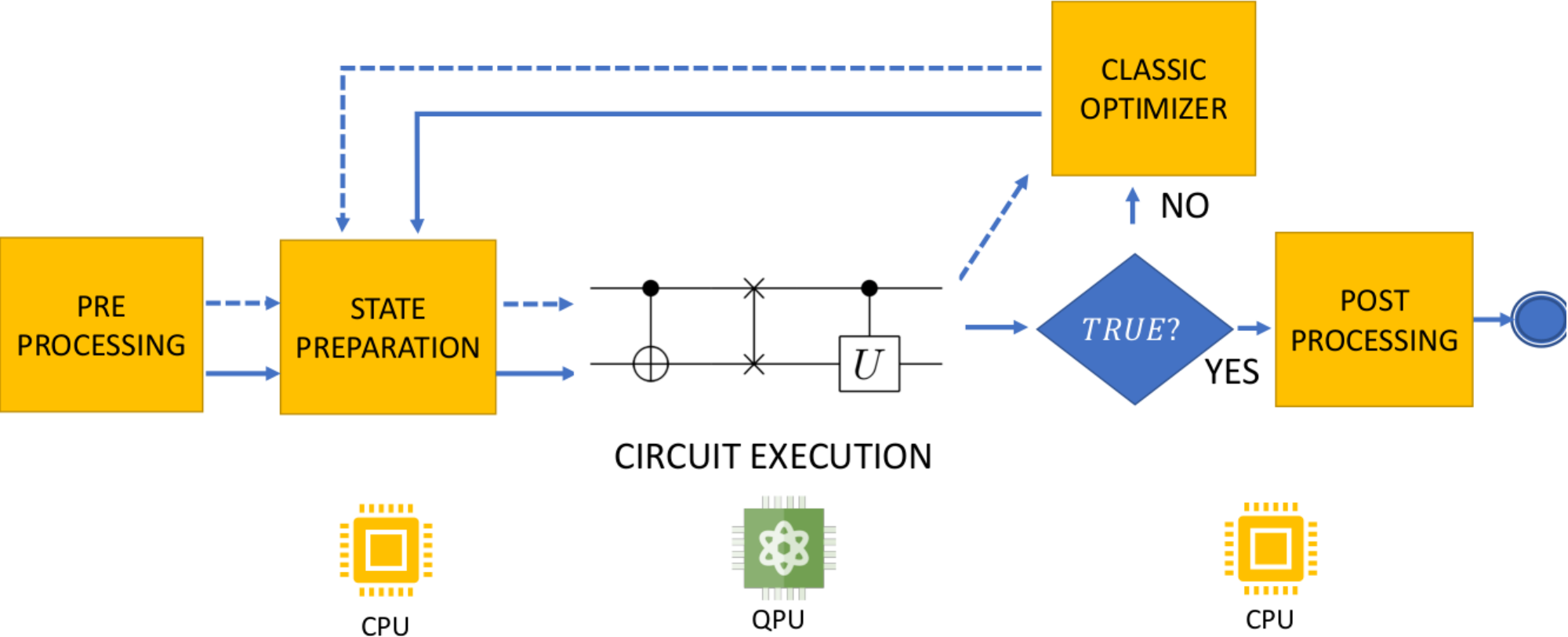}
    \caption{Hybrid Execution}
    \label{fig:hybrid-execution}
    \end{subfigure}
    \caption{Types of Quantum Tasks.}
    \label{fig:quantum-tasks}
\end{figure*}

We distinguish three types of quantum tasks, that are described in Figure~\ref{fig:quantum-tasks} together with their control and data flow: \emph{circuit execution}, \emph{task execution} (Figure~\ref{fig:task-execution}, and \emph{hybrid execution}. A circuit execution (Figure~\ref{fig:circuit-execution}) represents a single execution of a quantum circuit, where a single sampling of the result of execution is performed. Examples of this computation can be simple algorithms, such as the simulation of a coin toss. The result of a task execution, instead, is the most frequent result among $s$ samplings of the execution of a quantum circuit and represents the execution of typical quantum algorithms such as Grover and Deutsch-Jozsa. This model can also be applied to the quantum subroutines of hybrid algorithms, such as Shor's algorithm~\cite{365700}. Finally, a hybrid execution (Figure~\ref{fig:hybrid-execution}) represents computations involving interaction between classic and quantum hardware, similar to what happens with Variational Quantum Algorithms~\cite{cerezo2021variational}, where a quantum state modeling the solution to a specific problem is modeled as a parametrized quantum circuit with a vector of parameters $\vec{\Theta}$, whose optimal values are found by optimizers running on classic hardware. Finally, we define hybrid workflows as follows:

\begin{definition}[Hybrid Workflows]\label{Def:HybridWorkflows}
We define hybrid workflows as $W=(T,Q,E,D,f)$, where:
\begin{itemize}
    \item $T$ is the set of classic tasks, where $T' \subset T$ is the set of quantum candidates;
    \item $Q$ is the set of quantum tasks;
    \item $D$ is the set of decision nodes, where $|D| = |T'|$;
    \item $E \subset (T \cup Q \cup D) \times (T \cup Q \cup D)$ is the set of edges; 
    \item $f$ is the function mapping tasks in $T'$ in one or more quantum tasks in $Q$.
\end{itemize}
\end{definition}

\section{A Molecular Dynamics Use Case}
\label{sec:usecase}
MD is one of the most popular scientific applications executed on modern HPC systems. MD simulations reproduce the accurate dynamics  (time evolution of molecular systems) at a given pressure and temperature by iteratively computing interatomic forces and atom dynamics over short time steps. The trajectories generated by these simulations enable a better understanding of conformations and molecular mechanisms. In particular, a trajectory is a series of frames, i.e. sets of atomic coordinates stored at fixed infinitesimal time steps~\cite{Do2021}.

The quantum adaptation of MD applications builds over a similar pipeline as put forth in the recent works of~\cite{9041757, Do2021}; dealing with \emph{in-situ} and \emph{in-transit} analytics of MD simulations on state-of-the-art supercomputers. Our hybrid workflows can be efficiently integrated with the in-situ and in-transit analysis of
the data and meta-data generated by MD simulations. 
\medskip 

Here, we provide a high-level description of the hybrid workflows pertaining to the aforementioned use case. The motivation to focus on this particular use case is as follows: (a) MD simulations/analyses are a very active field in the distributed deep learning and HPC scientific computing communities with widespread applications to industry. (b) It consists of a purely quantum-based algorithm stack (see Section~\ref{sec:usecase:targetI}) followed by an additional variational hybrid algorithm stacks (VQEs) on the quantum algorithm returned output (see Section~\ref{sec:usecase:targetII}).  (c) It is elusive to find applications such as in~\citep{9041757, Do2021}, that offer one-to-one mapping of the entire classical problem (two major compute-intensive tasks) on a hybrid ecosystem and capture the essence of the hybrid workflow pipeline.
\medskip 

Our initial observation was that specific tasks of the \textit{in-situ} MD simulations, for e.g., collective variables (CVs) generation, can be leveraged using state-of-the-art quantum algorithms and quantum subroutines instead of merely parallelizing workloads using GPU accelerator facilities. Just in the way  GPUs, TPUs, and FPGAs serve as indispensable tools for accelerated linear algebra for variable dimension tensor calculations (matrix-matrix, matrix-vector, vector-vector multiply), the same logic also carries over to devices based on quantum architectures.  It is clear from the current status of quantum computers that a complete MD application cannot be executed on the rudimentary NISQ quantum hardware. Thus, we need to identify the best-suited parts (subspaces) of the workflows that can be accelerated multifold using quantum processing units (QPUs).  Moreover, the utilization time of quantum devices should be kept low-key, in order to prevent noise-induced quantum errors and imperfections of quantum hardware. 

\begin{figure*}[!h]
    \centering
    \begin{subfigure}{\textwidth}
        \centering
        \includegraphics[width=.65\textwidth]{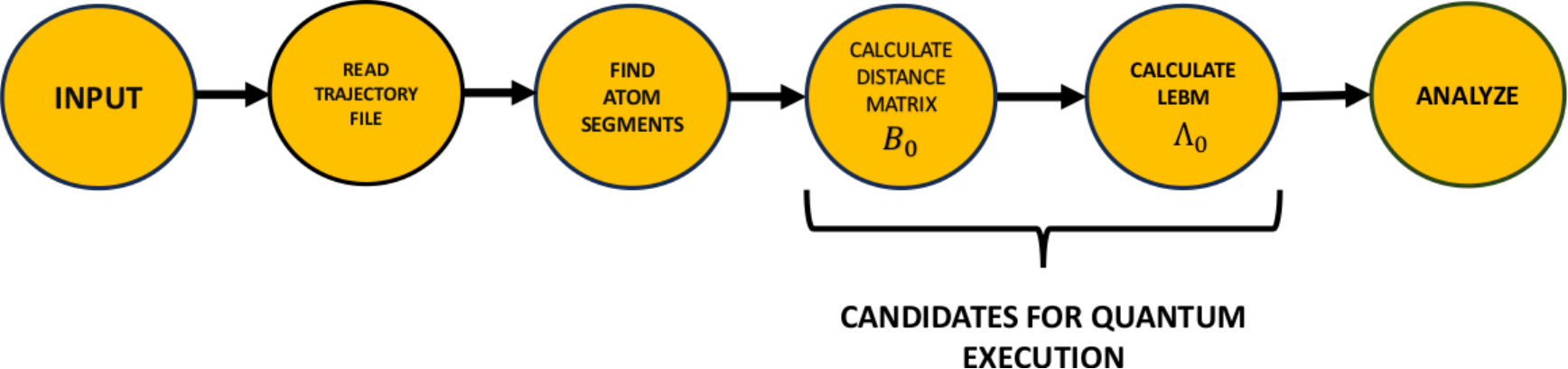}
        \caption{Target Classical MD Simulation.}
        \label{fig:classic-md}
    \end{subfigure}
    \begin{subfigure}{\textwidth}
        \centering
        \includegraphics[width=.65\textwidth]{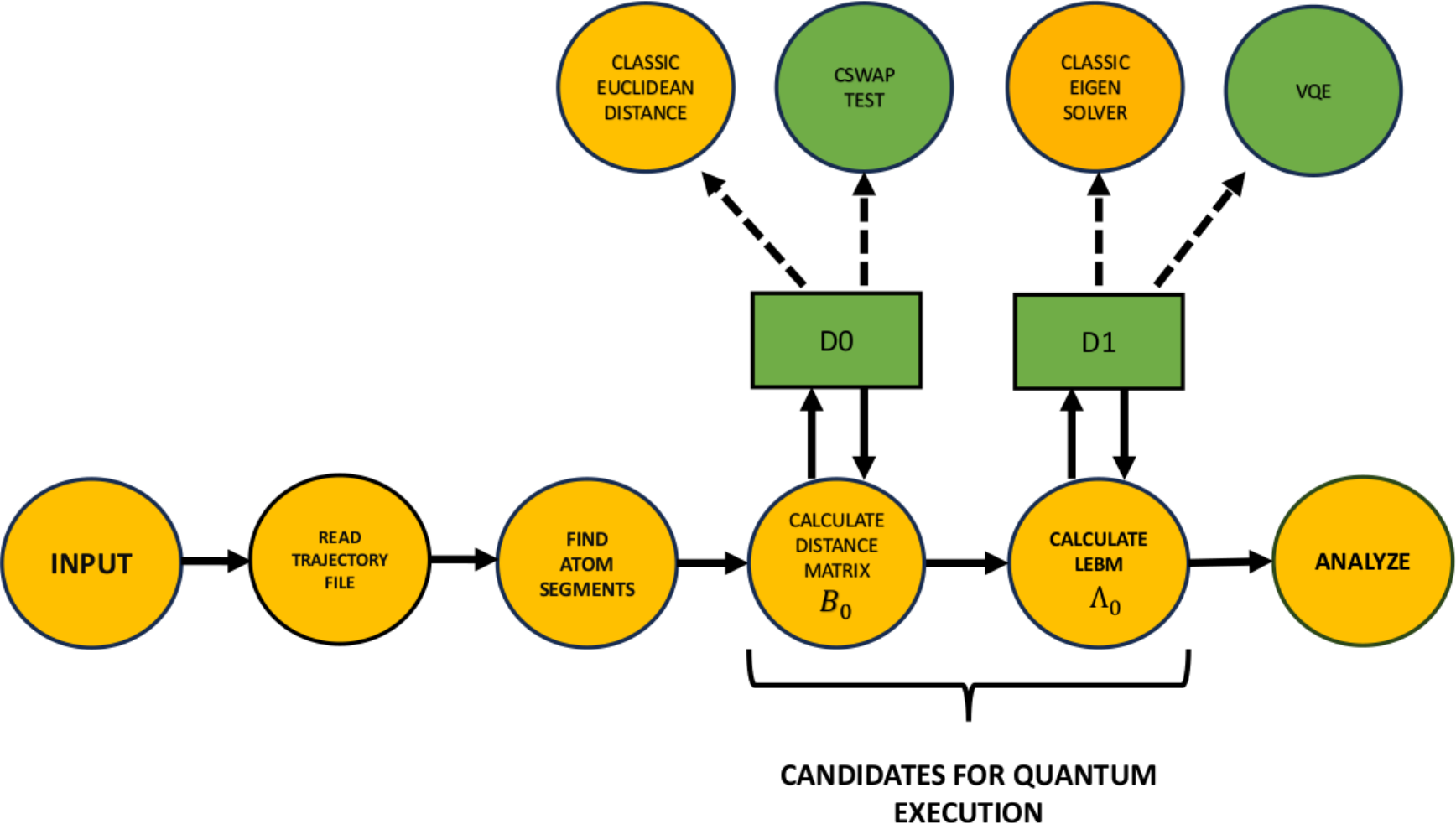}
        \caption{Hybrid MD Simulation.}
        \label{fig:hybrid-md}
    \end{subfigure}
    \caption{Hybrid Quantum-Classical MD Simulation.}
\end{figure*}
Figure~\ref{fig:classic-md} visualizes the target MD application with quantum candidates highlighted. After reading input from the user, the application reads a trajectory file, which provides information about the molecule structure, and identifies the atom segments that have to be considered in our calculations. For each pair of atom segments, the application performs parallel computations to calculate $B_{IJ}$ matrices. The CV that we will use in this work is the \emph{Largest Eigenvalue of the Bipartite Matrix (LEBM)}. The bipartite matrices are used as an input for the LEBM calculation. At the end of this process, results for different $B_{IJ}$ are collected and analyzed.

Starting from classical MD workflow (Figure~\ref{fig:classic-md}), by applying the procedure described in Figure~\ref{fig:hybrid-wfs} we obtain the resulting Hybrid Quantum-Classical workflow described in Figure~\ref{fig:hybrid-md}. These quantum-classical (hybrid) counterparts of the purely classical MD workflows were first conducted and benchmarked on a 5-qubit IBM Q devices~\cite{Cranganore2022}.  In the next sections, we describe each step in detail.
 
\subsection{Identification of Quantum candidates} 

Strategies for pinning down suitable quantum candidates in a use case is to a large extent problem/model/system dependent. In a hybrid workflow environment, this basically boils down to choosing from a limited set of available quantum and hybrid quantum-classical algorithms that could potentially outperform known classical algorithms.
\medskip 

In the MD simulation landscape, compute-intensive data analysis of time-evolving molecular systems can be systematically solved by designing a suite of collective variables (CVs)~\cite{9041757}. Collective variables (CVs) are a set of statistical metrics that capture relevant molecular motions enabling efficient monitoring of rare events in huge molecular structures and chains~\cite{Kitao1999-pm}.  Technically, CVs can also be defined as a function of the atomic coordinates in one frame that helps to reconstruct the free-energy surface for enhanced sampling.  Since trajectories are reduced to time series of a small number of such CVs, simulated molecular processes are much more amenable to interpretation and further analysis. Ideally, a CV can be as simple as the distance between two atoms or can involve complex mathematical operations on a large number of atoms. For example, in the domain of Metadynamics~\cite{barducci2011mtd}, scientists use well-chosen collective variables (CVs) to capture important molecular motions in the region of interest. 
\medskip 

 The works in~\citep{9041757, Do2021} demonstrated that collective variables could be extracted using Euclidean distance matrices and bipartite distance matrices. 

      


The Euclidean distance matrix \(D\) and the bipartite distance matrix \(B_{\text{IJ}}\) have three fundamental properties: they are symmetric, diagonal elements are zeros, and off-diagonal elements are strictly positive~\cite{9041757}.
After calculating matrices $D$ and $B$, we focus on calculating the molecular system's CVs. 
\medskip 

Potential quantum candidates for MD-based experiments were chosen in accordance to:

(i) Replacing classical routines that are compute-intensive with equivalent quantum algorithms that guarantee a theoretical speedup. A task $t \in T$ is defined as \textbf{compute-intensive} if at least $70$\% of its execution time is spent in performing floating point operations (i.e., no data staging, I/O, and communication).  

(ii) Capturing the electronic properties and molecular time-evolution of the system which in turn can be used for efficient computation of collective variables (CVs),       

(iii) Most importantly, whether there exists such a quantum algorithm and routine, that is available to the user for successful implementation of the use case. If not, then engineering quantum algorithms/quantum circuits with a considerably shallow circuit depth, (with a modest qubit register size and quantum gates crammed into the circuits) such that the problem can be efficiently simulated on NISQ hardware.
\medskip 

The two classical tasks that exactly satisfy all the aforementioned requirements in our use case are: 

(a) Distance (bipartite) matrix generation between different $C_{\alpha}$ atoms in the molecular system (cf. Target Task 1~\ref{sec:usecase:targetI}) for details), which is calculated using SWAP test;

(b) Calculating target CV corresponding to largest eigenvalues of the bipartite matrices (LEBM) and distance matrices using hybrid quantum algorithms (cf. Target Task 2~\ref{sec:usecase:targetII}) for details). We use Variational Quantum Eigensolvers (VQEs)~\cite{doi:10.34133/2020/1486935} for eigenvalue estimation.

\subsection{Mapping from Classic to Hybrid Workflow}\label{sec:usecase:targetI}

In Figure~\ref{fig:hybrid-md} we show the hybrid workflow resulting from the transformation of the classic MD workflow depicted in Figure~\ref{fig:classic-md}, following the transformation procedure described in Figure~\ref{fig:hybrid-wfs}. Since we have two quantum candidates, we add two decision nodes, $D0$ and $D1$. Decision nodes are responsible for selecting target implementation, depending on whether quantum hardware is available, and performing data encoding from classic to quantum domain. We analyze the implementation of each algorithm in the next sections, together with the control and data flow of target hybrid MD simulation.
\medskip 

\textbf{Target Task 1: Quantum algorithm for distance matrix generation \& computing structural change evolution}\label{sec:usecase:targetI}


It was shown in~\cite{https://doi.org/10.48550/arxiv.1307.0411} that quantum computers possess the power to manipulate large numbers of high-dimensional vector/tensor datasets. Typically, vector operations involving vector (tensor) dot products, norms, overlaps, etc., feature in supervised and unsupervised machine-learning tasks. For our MD analytics use case, CVs described by distance/bipartite matrices constitute some of the most compute-intensive routines. Classical algorithms running on classical devices typically require polynomial time in the number of vectors and the dimension of the space to solve these tasks. For an $N$-dimensional vector space, their time complexity grows linear in N, i.e., $\mathcal{O}(\mathtt{N})$. Quantum computers, on the contrary, can remarkably achieve this feat in time $\mathcal{O}(\mathtt{logN})$, owing to their intrinsic capabilities to efficiently manipulate high-dimensional vectors embedded in large tensor product spaces~\cite{https://doi.org/10.48550/arxiv.1307.0411}.
\medskip 

The reason for this exponential speed-up can be argued as follows: classical data (typically vectors or tensors), expressed in terms of $N$-dimensional complex-valued vectors can be encoded through amplitude encoding onto merely $\mathtt{log_2 (N)}$ qubits, thus requiring only logarithmic in the size of the classical data-set. This data stored in a \emph{quantum random access memory} (\texttt{qRAM}) whose mapping takes $\mathcal{O}(\mathtt{log_2 N})$ steps~\cite{PhysRevLett.100.160501}. In this converted quantum form, data post-processing can be done using multiple quantum algorithms like the \emph{quantum Fourier transforms}~\cite{nielsen_chuang_2010}, matrix inversion methods~\cite{PhysRevLett.103.150502}, etc., with time-
\newline 
complexity $\mathcal{O}(\mathtt{poly(logN)})$. Thus, distance estimation and inner product operations between post-processed vectors belonging to N-dimensional vector spaces take merely time $\mathcal{O}(\mathtt{logN})$
\newline 
~\cite{https://doi.org/10.48550/arxiv.1307.0411}. Moreover, as per~\cite{https://doi.org/10.48550/arxiv.0910.4698}, sampling post-processed vectors and distance or inner-product determination between these post-processed vectors is an exponentially hard task.

Inspired by the aforementioned (theoretical)  \textit{exponential} 
\newline 
quantum speedup, we make use of quantum algorithms to generate the proxies for modeling molecular structure time-evolution, i.e., Euclidean distance matrices $D$ and bipartite distance matrices $B_{\text{IJ}}$. These can be efficiently generated via a quantum sub-routine called the \emph{SWAP} (also called \emph{C-SWAP}) test.

\medskip 

The high-level workflow description of the \textbf{Target Task 1} with quantum integrated architectures is depicted in Figure (\ref{fig:quantum-CV}). We can see that its execution resembles the control and data flow of a job execution, as shown in Figure~\ref{fig:task-execution}. Typically, it requires manual re-direction of tasks onto quantum devices for generating distance matrices. The proposed workflow schematics is a quantum adaptation (counterpart) of the \emph{in situ} or \emph{in transit} integrated classical workflows used for molecular dynamics (MD) analytics, originally introduced in~\cite{https://doi.org/10.1002/jcc.24729}.

\begin{figure*}[!ht]
    \centering
    \includegraphics [width=\textwidth]{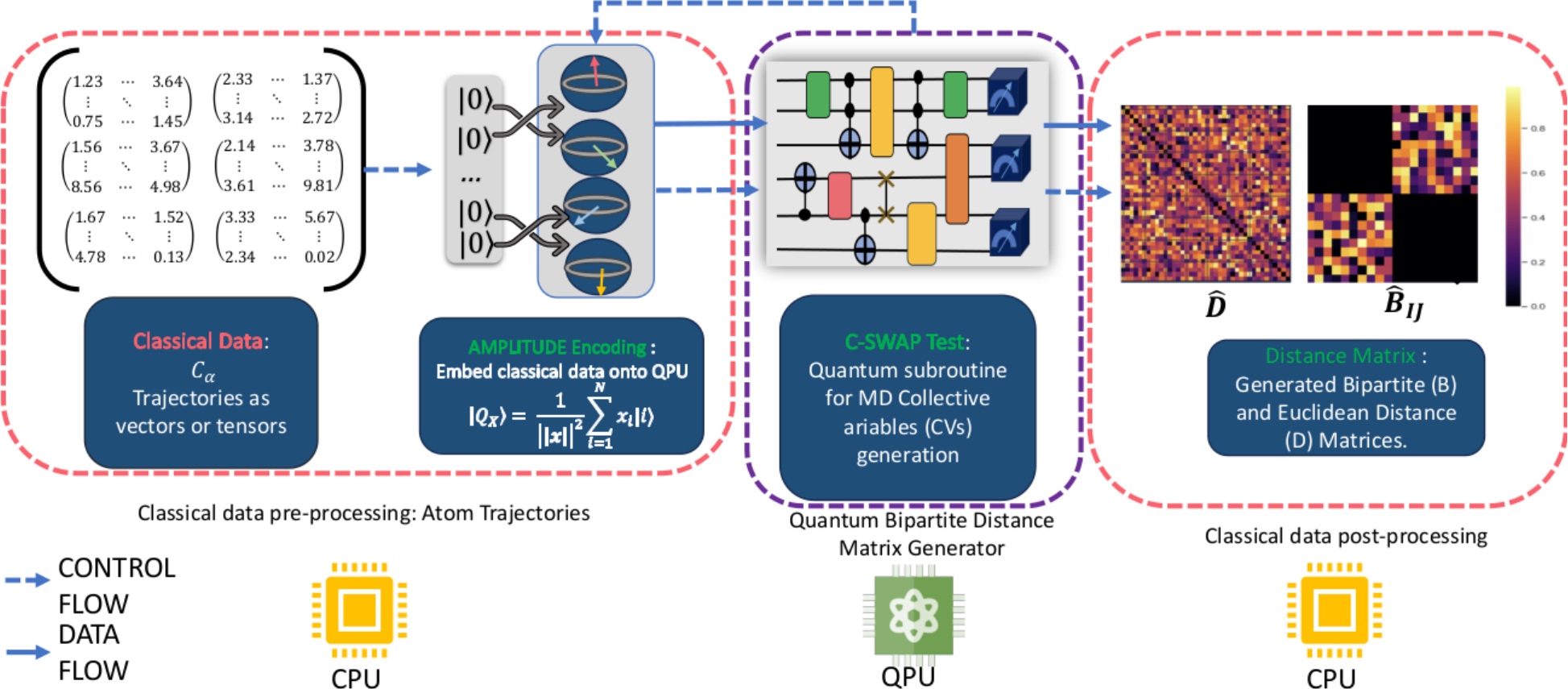}
    \caption{Quantum integrated MD workflows for distance matrix computations.}
    \label{fig:quantum-CV}
\end{figure*}
\medskip 

\textbf{Target Task 2}: \textbf{Largest Eigenvalues of Bipartite Matrix}\label{sec:usecase:targetII}
\textit{Hybrid Quantum-Classical algorithm for Eigenvalue Determination}:
It was first posited in the works of Johnston {\it et al.}~\citep{johnston2017isd, https://doi.org/10.1002/jcc.24729} that the measurement of largest eigenvalues of the bipartite matrices $B_{\text{IJ}}$ or the Euclidean distance matrices $D_{I}$ can serve as collective variables (CVs) and suffices for monitoring structural changes in the conformation of $I$ relative to $J$. Later, this idea was extrapolated in~\citep{Do2021, 9041757}, wherein the largest eigenvalues of each one of these matrices were computed without retaining other frames in memory. Nevertheless, with ever-increasing molecule sizes, numerical linear algebra calculations such as eigenvalue and eigenspectrum (typically using numerical diagonalization techniques) or singular value decomposition (SVD) computations may fall short due to the saturating computational power of classical (HPC) systems.

Hybrid quantum-classical systems offer a possibility to alleviate the aforementioned bottleneck by semi-utilizing quantum devices as accelerators orchestrated by classical optimization protocols for parameter updates. The \emph{variational quantum eigensolver} (VQE)~\citep{Peruzzo2014-uw, Tilly2022-fv, McClean_2016} form an important subset of the variational quantum algorithms (VQAs), that computes the \emph{eigenvalues} of typically large Hermitian matrices $\hat{A}$ (see Appendix F) using the \textit{Rayleigh-Ritz} variational approach~\cite{Ritz1909}. This heuristic algorithm was developed with a strong focus on solving the ground state of many-body interacting quantum systems (strongly correlated) using iterative numerical optimization. Multiple highly complex systems appearing in quantum chemistry remain in-tractable even for the capabilities of current leading-edge high-performance computing (HPC) systems. Augmenting NISQ devices or the near-future quantum devices to operational support given by supercomputers increases the hopes for a faster convergence to solutions for such large-scale chemistry target applications based on quantum simulations. ~\citep{cerezo2021variational, broughton2021tensorflow, Babbush2019}


The LEBM computation of the target matrices $B_{\text{IJ}}$ or $D$ falls as a fitting use-case for the variational quantum eigensolver engine. The machinery is described in detail in Appendix F. Initially one begins with an ansatz wavefunction (trial state) outputted from a parameterized quantum circuit. The core principle operating under the hood is iteratively updating the wavefunction parameter whilst minimizing the \emph{cost function} Eq.(\ref{eq:background:vqa:cost}) $\mathcal{C}(\boldsymbol{\theta})$ by employing a classical optimizer. This cost function is typically chosen to be the quantum expectation value of a given hermitian matrix with respect to the parametrized wavefunction $\ket{\Psi(\boldsymbol{\theta)}}$. The input to the VQE engine is either the C-SWAP quantum-subroutine generated bipartite matrix or generated bipartite matrices generated using classical machines.  (which is Pauli encoded, since it is a Hermitian in nature cf.~Appendix E), the corresponding function to minimize reads, 

\begin{equation}
    \vartheta^* = \operatorname*{arg\,min}_{\boldsymbol{\theta}} C(\boldsymbol{\theta}) = \bra{\psi(\boldsymbol{\theta})}\hat{B}_{\text{IJ}}\ket{\psi(\boldsymbol{\theta})},
    \label{eq:background:vqa:cost}
\end{equation}
\medskip 
The goal is to approach the sets of LEBM as close to the actual values, i.e., the values calculated on a classic machine. However, since VQE performs iterative optimization of a cost function $C$ (Equation~\ref{eq:background:vqa:cost}), it can approach the exact value only after multiple executions or iterative loops~\cite{cerezo2021variational}. To this end, we define the error function that we use to quantify the VQE benchmarks.

Our quantum-enhanced MD simulations were conducted on a $5$ qubit IBMQ hardware. The maximum permissible matrix dimension possible with our requested quantum resources was a restrictive $16 \times 16$ dimensional distance matrix blocks $E_{\text{IJ}}$ that feature in the bipartite matrix.  
The corresponding \emph{VQE scientific workflow} for molecular-dynamics (MD) target applications (cf. Appendix F for detailed mathematical description) is depicted in Figure~\ref{fig:VQE workflow}, which resembles hybrid execution on Figure~\ref{fig:hybrid-execution}.
\begin{figure*}[!ht]
    \centering
    \includegraphics [width=\textwidth]{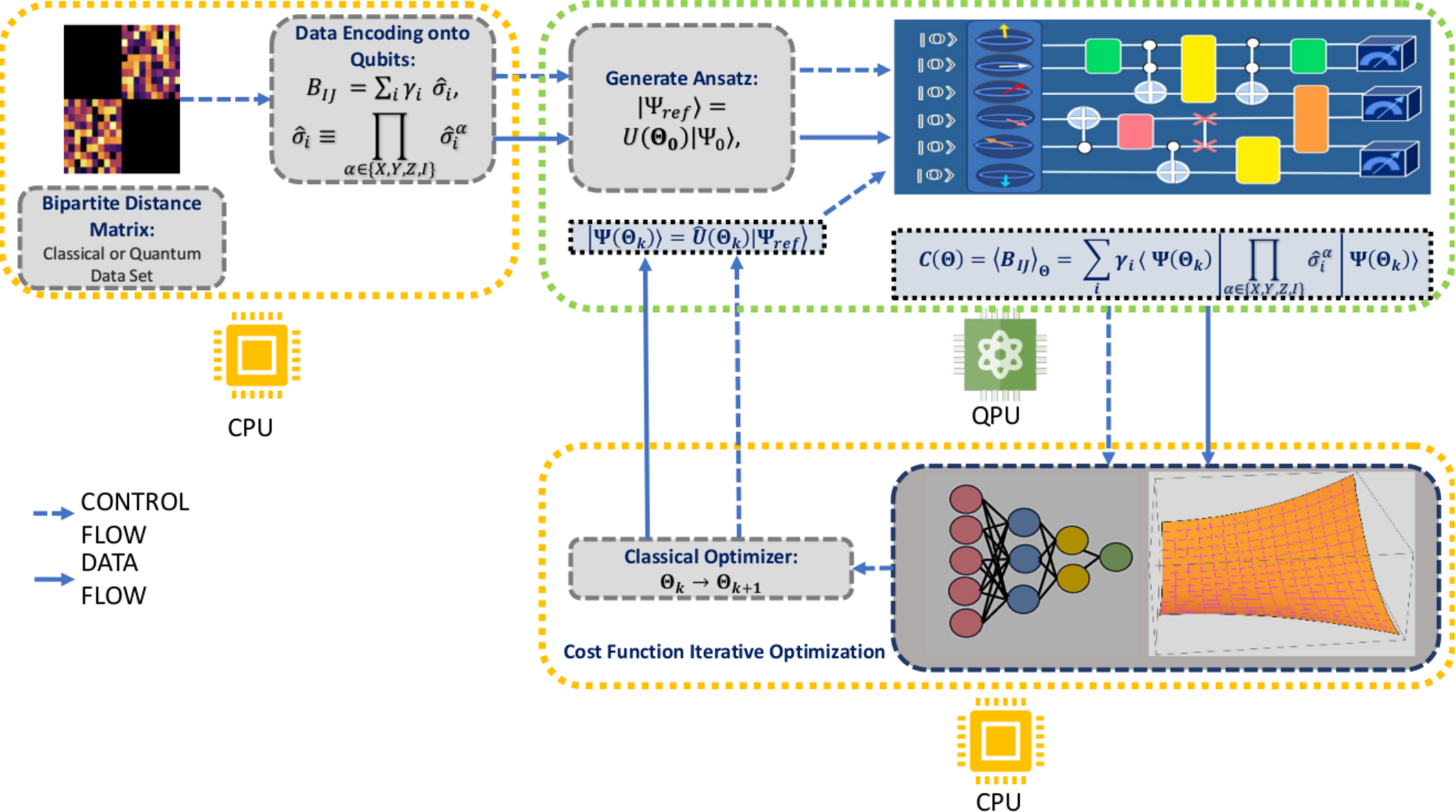}
    \caption{Quantum workflows for variational quantum eigensolvers}
    \label{fig:VQE workflow}
\end{figure*}


\paragraph{Cost Function Optimization} 

For the $\hat{B}_{\text{IJ}}$, we define its classic LEBM $\Lambda_{c}(B_{IJ})$ and its quantum counterpart calculated using VQE using a specific initially chosen hyperparameter setting $\Pi$ as $\Lambda_{vqe}(\hat{B}_{\text{IJ}},\Pi)$. We construct a figure of merit to compare VQE results obtained from different architectures, quantified by a Mean Square Error (MSE). For a set of matrices,  $\hat{B}=\{B_{IJ}^0, B_{IJ}^1, \ldots, B_{IJ}^n\}$, we define 
\begin{equation}
Err(\hat{B}, \Pi) = \sum_{i \in [0, |\hat{B}|]} \frac{(\Lambda_{c}(B_{IJ}^{i}) - \Lambda_{vqe}(\hat{B}_{\text{IJ}}^{i},\Pi))^2}{|\hat{B}|}
\label{eq:objective:error}
\end{equation}
as the MSE between the classic and quantum eigenvalues, calculated using hyperparameters $\Pi$.

\paragraph{Data-Driven Hyperparameters Selection} 
It is important to note that the VQE executions can be largely improved by opting for different hyperparameter $\Pi$ tuning schemes that are available on the IBM Qiskit API. In practice, VQE execution results are affected by tuning hyperparameters and converge to the optimal value after the user pins down a suitable configuration after taking into account possible error-mitigation strategies. Therefore, a grid-search-based algorithm forms an important part of the so-called data-driven methods which will be used to identify the most suitable hyperparameters setting. In~\cite{Cranganore2022}, we describe a method to select a suitable set of hyperparameters for VQE.

\subsection{Generalizing Hybrid WMSs}

The rapid progress of quantum facilities has led to a surge in challenges across various application domains. Consequently, generalizing the hybrid workflow frameworks to accommodate diverse applications is necessary. A comprehensive overview is depicted in Fig. (2), which could encapsulate a broad spectrum of techniques and be embedded at various stages in a generalized hybrid ecosystem. These encompass (a) Hamiltonian simulation and analog/single ancilla Linear Combination of Unitaries (LCU)~\cite{2012} methodology for ground state preparation, (b) Graph-based combinatorial optimization employing quantum-approximate-optimization algorithms (QAOA)~\cite{farhi2014quantum}, (c) Quantum Fourier Transformations (QFT) and Quantum Phase Estimation techniques (QPE) as alternatives to Variational Quantum Eigensolvers (VQEs), (d) Classical machine learning (ML) assisting quantum algorithms~\cite{10.5555/2871393.2871400}, (e) Quantum machine learning (QML) workloads such as quantum support vector machines (QSVMs)~\cite{PhysRevLett.113.130503}, Quantum Kernel-based ML (QKE) methods ~\citep{liu2021rigorous, Havlicek2019-nk}, and (f) Simulations of quantum Hamiltonians for determining electronic state energies~\citep{Clinton2024-mf, drugquantum}. The following sections build over this generalized scheme.

\subsubsection{Hybrid Data/Control Flow}
\begin{figure}[!h]
    \centering
    \includegraphics[width=1.0\columnwidth]{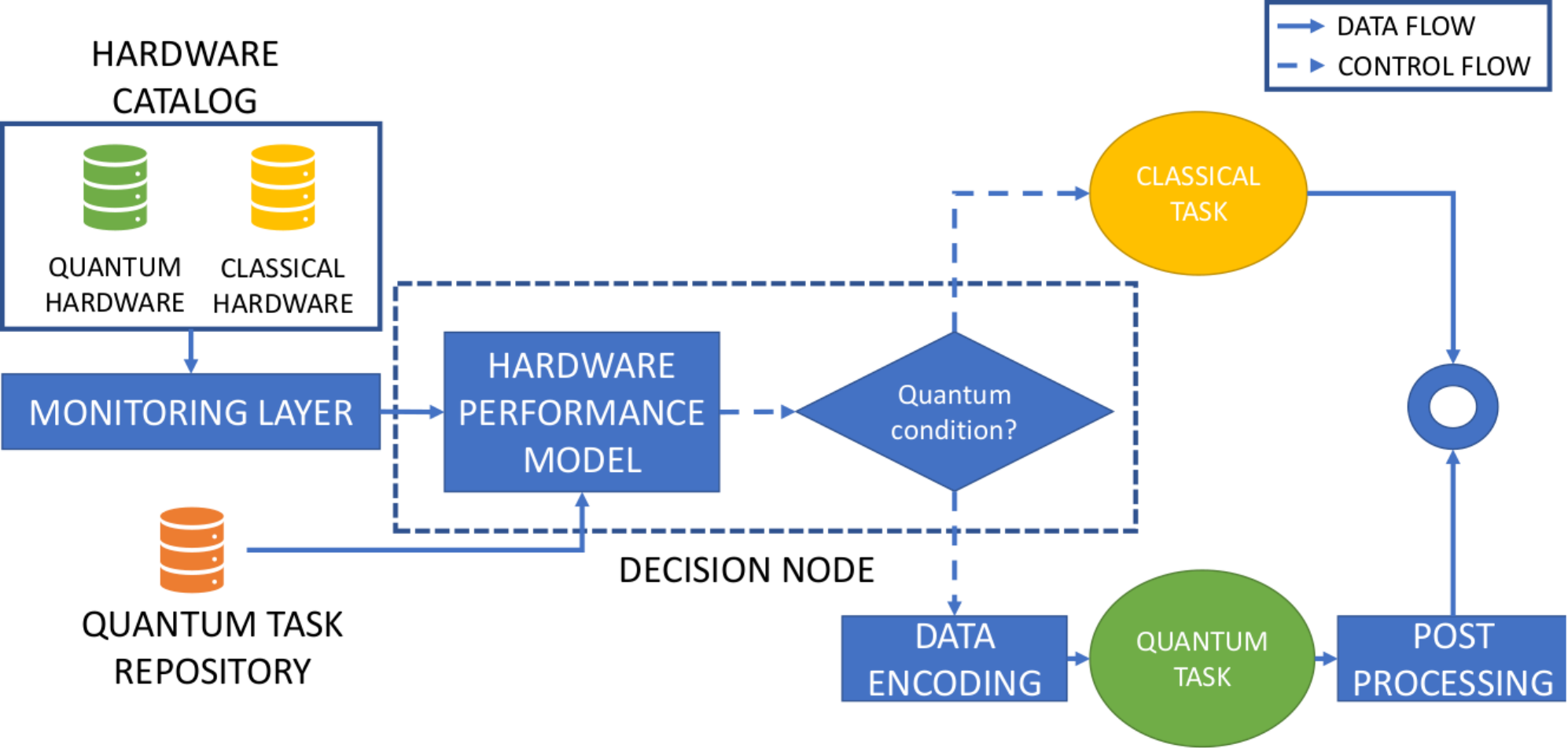}
    \caption{Decision Node.}
    \label{fig:decision-node}
\end{figure}

In a Hybrid Workflow Engine (HWE), characterizing the control and data flows becomes highly relevant. Since the control flows (defining execution paths) and data flows (information movement through computations) are highly interlinked, it is necessary to demarcate and differentiate these two flows systematically. In such intertwined workload sharing environments arising in \emph{variational quantum algorithms} between 
\newline 
quantum-classical components, WMSs become highly non-trivial. Here we relate/map the conducted quantum MD computational experiments to the formally defined hybrid workflows from Section~\ref{sec:hy}.

The data and the control flow are handled by the decision nodes, whose flow is described in Figure~\ref{fig:decision-node}. The goal of decision nodes is to manage the control and data flow of the hybrid workflow, based on pre-defined conditions for execution on quantum hardware. Decision nodes collect data about the available quantum and classical hardware through the monitoring layer. Also, available quantum tasks are fetched by the quantum task repository, following indications provided by annotations. Afterward, performance is evaluated for each task on the available hardware, by applying a specific hardware performance model. Performance models can be either defined as a single metric, i.e., quantum volume~\cite{cross2019validating} or total quantum factor~\cite{sete2016functional}, or exploit machine learning-based approaches such as~\cite{quetschlich2023mqt}.   Based on the results of the hardware performance models, a specific condition is evaluated to make a decision. Based on the condition, the classic or the quantum task is executed. If the quantum task is executed, classic data has to be encoded into the quantum domain by applying different data encoding methods. After execution on quantum hardware, post-processing is applied to mitigate the effects of quantum noise. Finally, the output of the task, either classic or quantum, can be used by the subsequent workflow task. 

Similar approaches to the selection of quantum tasks have been proposed. For example,~\cite{Salm2020} proposes NISQAnalyzer, a method for selecting the most suitable quantum machine among a set of available superconducting quantum hardware and Qiskit implementations, while~\cite{Garcia-Alonso2022} focuses on the execution of quantum services on Amazon BraKet. In~\cite{quetschlich2023mqt}, these ideas are further explored considering also ion-traps architecture and different programming frameworks. However, different hardware, for the same algorithm might require different implementations and different programming models. An example is photonics hardware, where Measurement-Based Quantum Computing~\cite{briegel2009measurement} provides better performance~\cite{zhang2023oneq}. Therefore, we propose an approach where a joint selection of algorithm implementation and target quantum hardware is performed. 

In the next sections, we analyze control and data flow in detail.

\paragraph{\textbf{Control-Flow Modeling for Hybrid Systems}}
Control-flow modeling deals with deconstructing the order of operations in a computer program and dictates the sequence in which computational tasks need to be executed~\cite{9514390}.  In the case of hybrid quantum-classical applications, execution paths can be bifurcated into the classical control flow pipeline (pertaining to classical flows and smaller sub-flows) and quantum control flow pipeline (related to quantum flows and sub-flows) respectively.

\textbf{Classical side of quantum computation}: Typically, the classical control flow dependencies involve:
\begin{enumerate}
     
    \item Initializing set of variational parameters of the quantum circuit.

    \item Spawning multiple quantum circuits and processes simultaneously.
    
    \item Delegating (directing) user-defined quantum tasks $Q$ onto the quantum processors and controlling the execution order of one or several quantum candidates. 
    \item Executing classical programs in concert with quantum routines (e.g., for iterative optimization algorithms). 
    \item Measurement of quantum states and classical parameter feedback. Lastly, generating the final output after multiple iterations (checking convergence to solution) and data post-processing.
\end{enumerate}


\paragraph{Data-Flow Modeling for Hybrid Systems}
 Data-flow analysis on the other hand deals with collecting information about the possible set of values calculated at various stages in a computer program and describes the information passage within a workflow~\cite{9514390}. Data generation, feature extraction, data transformation, data storage, and exchange of input and output data within the workflows are some examples.  
\medskip 
 
\textbf{Classical Data flow modeling:}
\begin{enumerate}
    \item Preprocessing input classical data, implementing feature extraction methods on classical datasets, and developing suitable forms of data encoding schemes for QPU embedding. 

    \item Postprocessing results after quantum measurements, mitigating readout errors and extracting useful information.
\end{enumerate}

\textbf{Quantum Data flow modeling:}
\begin{enumerate}

    \item Quantum data flows concerns with quantum parallelism arising due to multi-qubit states $\ket{\Psi} \in (\mathbb{C}^2)^{\otimes d}$ within a quantum register and the flow of entanglement prepared via controlled gate sequence applications. (for e.g., multi-QPU interactions and managing multi-QPU communication). This may be necessary to communicate a computational state between QPUs or to prepare both registers in a mutually entangled state. These operations require the presence of a quantum network, which uses quantum physical systems to communicate quantum states between registers)~\cite{10.1145/3007651}.

    \item Exchanging data within the several spawned quantum circuits or subparts of the different circuits.

\end{enumerate}

\section{Hybrid Workflow Management Systems}
\label{sec:hybrid-wms}
\subsection{Software Components}
In this section, we describe the main software components that are necessary to enable the execution of tasks into hybrid quantum-classical systems. The main components are also depicted in Figure~\ref{fig:quantum-workflows}.

\textbf{Hardware Catalog} contains information about available hardware in underlying hybrid computing systems, including not only which hardware is available, but also different characteristics (i.e., CPU power, storage capacity, network throughput). This component is common in many other WMSs, such as Pegasus~\cite{pegasus-ref}. Hardware-specific implementations of the hardware catalog are available in different quantum hardware with Cloud frontends, i.e., IBM Quantum, Amazon BraKet, Google Quantum AI, and Azure Quantum. The main difference between these services is that, while IBM and Google have their own quantum devices, Amazon provides an interface to access quantum devices from other vendors, e.g., Rigetti, Quantinuum, and Pasqal. Azure Quantum, instead, proposes a hybrid solution, exposing the same interface not only for Microsoft devices but also for devices of other vendors. The first step towards the integration of quantum hardware in the hardware catalog would first of all require the design of APIs to interconnect these services with WMS hardware catalog. Existing approaches described in~\cite{Salm2020,Garcia-Alonso2022,quetschlich2023mqt} are currently not integrated into WMS, targeting only specific frameworks (e.g., Qiskit and IBM Quantum~\cite{Salm2020} or Amazon BraKet~\cite{Garcia-Alonso2022}) or specific hardware, e.g., superconducting and ion-traps~\cite{quetschlich2023mqt}. Also, information about the characteristics of available hardware, also known as quantum hardware descriptors, should be exposed. Typical descriptors are a number of available qubits, qubits topology, and error rate. Moreover, since quantum computers at the time could not be used concurrently, information about the queueing status should be exposed. In some cases, such as IBM, also other performance metrics that are specific to quantum, i.e., Circuit Layer Operations per Second (CLOPS) or Quantum Volume~\cite{cross2019validating}, can be exposed. However, such metrics might not be directly applicable to other types of hardware, i.e., from other vendors or relying on different technologies, such as, for example, ion-traps, photonics, or neutral atoms. 

\textbf{Quantum Task Repository} includes the implementation of different quantum tasks, in different software variants depending on the available quantum machines. Tasks can be implemented in graphical languages, such as ZX-Calculus~\cite{van2020zx}, different high-level programming frameworks (e.g., Qiskit, PennyLane, Q\#), or in low-level languages such as Open\emph{QASM} (quantum assembler)\footnote{https://github.com/openqasm/openqasm}.  Examples of quantum task repositories are provided by the IBM Quantum Lab, BraKet git repository\footnote{https://github.com/amazon-braket}, Cirq Quantumlib github\footnote{https://github.com/quantumlib/Cirq}, Microsoft Quantum github\footnote{https://github.com/microsoft/Quantum}. However, each one of these repositories targets specific frameworks and devices. A generic quantum task repository should include implementations for different devices available in the hardware catalog. Also, each task should be labeled to describe its goal (i.e., unstructured search, optimization), as well as its input and output. Compilation of selected tasks will then be performed in the subsequent transpilation layer.

\textbf{Classic-Quantum Mapper} is responsible for identifying a mapping of classic tasks into equivalent quantum tasks. This can be performed by means of \emph{code classification}~\cite{ohashi2019convolutional,chen2022overview}, using labels defined in the Quantum Task Repository, or by applying circuit synthesis methods, such as~\cite{Davis2020TowardsOT,atkinson2019quantum} if no corresponding quantum task is found. Currently, the mapping of quantum candidates into quantum tasks is performed manually by the application developer. One could think of exploiting code classification approaches~\cite{ohashi2019convolutional} 

\textbf{Transpilation Layer} is responsible for transpile, i.e., adapting the high-level definition of quantum circuits that define the quantum tasks in the quantum tasks repository, to the target quantum architecture. This step is necessary because the definition of the circuit might not fit the topology of the underlying architecture. To address this issue, a sequence of operations is applied to the circuit definition to reorganize qubits and quantum gates. Transpilers are available in different frameworks, such as Qiskit, PennyLane, and Q\#. However, they are mostly designed for circuit-based quantum computing. Work available for measurement-based quantum computing~\cite{zilk2022compiler}, that allows to fully exploit photonics quantum devices is at the moment still at the experimental phase. In the future, we should be able to have universal transpilers, capable of optimizing for different computational models depending on available quantum hardware.

\textbf{Hybrid Monitoring Layer} is designed to collect data about the execution of both quantum and classic hardware. On classic hardware, there are many methods to collect either low-level metrics about the systems (e.g., CPU, bandwidth) or aggregated metrics (e.g. reliability), ranging from distributed~\cite{zhang2022elastic} and centralized~\cite{sanchez2016design} monitoring approaches. To enable monitoring of hybrid quantum-classical systems, the existing monitoring layer provided by existing WMS should be integrated with API provided by existing quantum framework (i.e., Amazon BraKet, IBM Quantum, Xanadu) to expose to the workflow developer the hardware descriptors collected by quantum computers, described in the hardware catalog. The implementation of the monitoring service for each provider is related to the underlying framework: for example, in IBM Quantum, it is provided by the runtime service, that allows keeping track of the state of a submitted quantum job by querying the state of a Job object. Similar approaches are available on PennyLane and Q\#.

\textbf{Hybrid Intercommunication Layer} constitutes the API that allows communication between quantum and classic hardware. This layer allows data encoding, offloading of data and computation to quantum hardware, retrieving measurements, and performing error correction. Currently, such features are framework-specific and therefore depend on the available API between each framework and the hardware vendors.

\subsection{Hybrid Workflows Execution}
Figure~\ref{fig:quantum-workflows} summarizes how we envision the execution of hybrid workflows. First, the user submits a scientific workflow using the WMS interface. Workflow tasks are then selected by the WMS scheduler, based on its scheduling policy. Based on the workflows annotations, the scheduler knows if the current $t$ task is a \textbf{classic} or a \textbf{quantum} task. If $t$ is classic, the scheduler will select target machine based on (1) task requirements, (2) availability of hardware resources (based on input hardware catalog), and (3) task-related cost function, which takes in input task $t$ and machine $m$ and outputs a score, which defines whether it is convenient to execute $t$ on machine $m$ and it is used by the scheduler for its decisions. Execution is then performed on the identified machine $m$, and the results of task $t$ are forwarded to the user or to the following tasks that require them as input. 

If $t$ is a quantum task, the first thing to do is to verify whether a quantum \textbf{target}, i.e., an equivalent quantum algorithm for the classic quantum candidate, is available in WMS quantum codebase. If no equivalent quantum target is found, the workflows continue its execution by scheduling the classic task on the available classic hardware; otherwise, according to the logic implemented in the decision node, the WMS can decide either to (1) select a quantum algorithm among the available quantum targets, or (2) execute the classic implementation of the task. In the first case, execution depends on the quantum task type (as defined in Figure~\ref{fig:task-execution}), in the latter, execution proceeds as in classical workflows.

\begin{figure*}[!ht]
    \centering
    \includegraphics[width=0.9\textwidth]{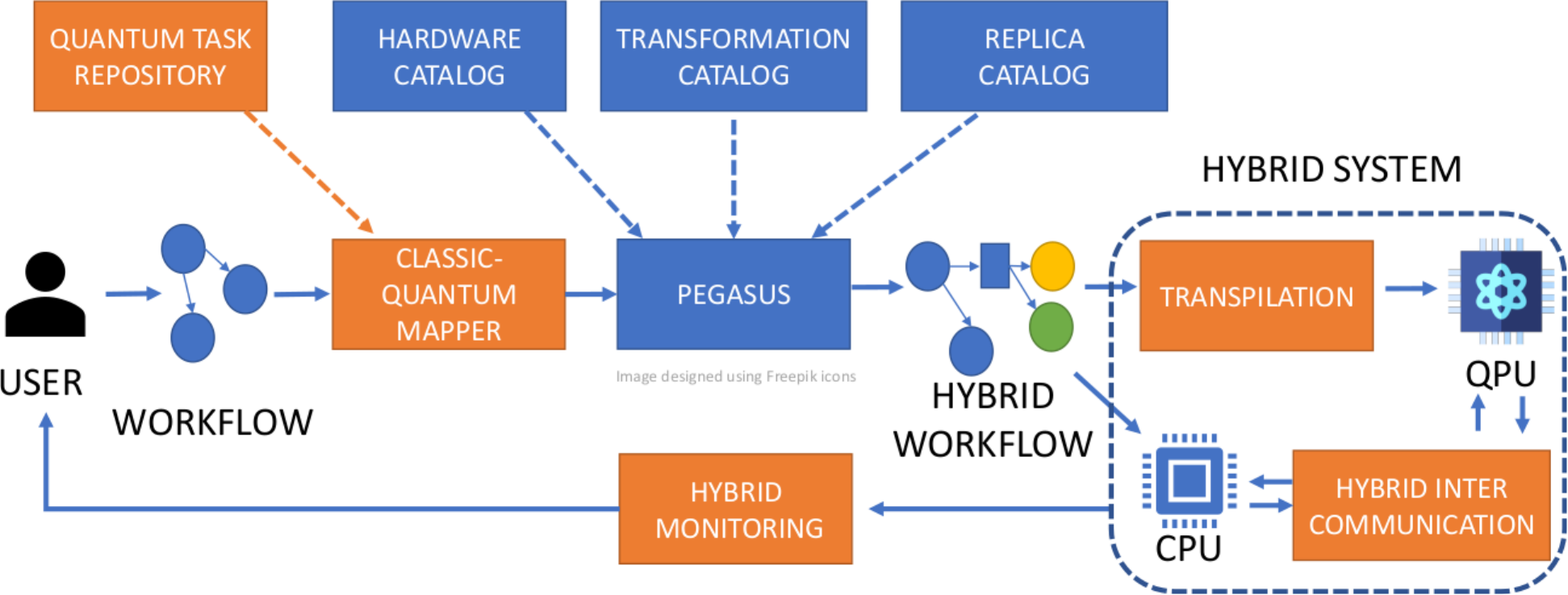}
    \caption{Quantum Workflow Management System.}
    \label{fig:quantum-workflows}
\end{figure*}
\subsection{Assumptions and Limitations}
Description of quantum task repository, as well as transpilation as mapper, are focused on the quantum circuit model, that is typical of superconducting machines used by the most common quantum hardware platforms (e.g., IBM Quantum, Amazon BraKet) and also by ion-trapped machines (e.g., AQT). However, different technologies are available for quantum hardware, such as photonics quantum machines (e.g., Xanadu). While each circuit can be also executed on photonics hardware, these machines proved to have the best performance with measurement-based model~\cite{briegel2009measurement}, therefore circuit synthesis and mapping might not guarantee the best performance.

Another issue is related to the code classification-based approach since we assume that there is a label for each classic task that allows us to map it to its quantum counterpart. As a consequence, this approach is inherently dependent on the source code corpus that we use to train the model. 

Concerning hardware descriptors, we mostly focus on definitions provided by IBM Qiskit, such as quantum volume and CLOPS, whose definition is based on superconducting IBM machines. While these metrics could be extracted also for other types of hardware employing the circuit-based model, they might not be applicable to annealers or photonics.

\section{Challenges}
\label{sec:challenges}
Based on our vision of execution of hybrid workflows, we identify the challenges that must be addressed to realize our vision of hybrid classic/quantum workflows.

\subsection{Quantum Hardware Descriptors}
As already mentioned in Section~\ref{sec:hybrid-wms}, hardware descriptors should be part of the hardware catalog, to characterize different available quantum hardware and to perform decisions on whether to allocate. The first challenge is to identify descriptors that can be used to describe quantum hardware: while on classic CPUs we have well-known descriptors (e.g., CPU frequency, RAM frequency, and memory, available storage), quantum hardware cannot be characterized using the same descriptors, due to the inherently different computational model. Also, there are different hardware technologies available for quantum computers (e.g., superconducting, neutral-atoms, ion-traps, photonics, and annealers), which makes it difficult to identify the descriptors most impacting the performance of quantum computing. 

To address this challenge, first of all, in-depth empirical studies are required to identify the most impacting parameters. Empirical studies require also the definition of purposely designed benchmarks, similar to~\cite{Arute2019,Willsch2022,jałowiecki2023pyqbench}. Based on data collected from the benchmarking, we can identify a subset of descriptors that is of interest for a specific quantum hardware. Descriptors can be also composed of aggregated hardware parameters, such as quantum volume~\cite{cross2019validating}, that is used by IBM Qiskit. 

Moreover, applications need to be able to target different QPUs automatically depending on the problem type, involving minimal code modifications. Hence, a \emph{QPU agnostic} development allows greater freedom and efficiency in hybdrid ecosystems. Recently, a framework named \emph{MQT Predictor}~\cite{quetschlich2023mqt} that automates quantum device selection and according to compilation of quantum algorithms for suitable hardware, execution has been proposed. This toolkit\footnote{https://github.com/cda-tum/mqt-predictor} offers a) prediction method, based on Supervised Machine Learning, which without performing any compilation predicts the most suitable hardware specific to the application in consideration. b) A reinforcement learning (RL) based method producing device-specific quantum circuit compilers. The compilation passes from several compiler tools are combined by learning (trained) optimized sequences of those passes (following a mix-and-match compiler pass) with respect to a customizable fidelity.

\subsection{Performance Models}
Performance models should be able to predict different performance metrics (i.e., running time, error rate) taking as input (1) a quantum task, and (2) a target quantum hardware where the input task should be executed. Such performance models are also required by WMS schedulers to decide where to allocate input tasks. Performance models should be based on values of identified hardware descriptors to facilitate the integration of different architectures. The main challenges in designing performance models lie in the difference between quantum architectures, which makes it difficult to design a generalized model that can be applied to different hardware. Also, since quantum error variates with time, depending on different environmental factors~\cite{Preskill2018quantumcomputingin}, performance models should be automatically updated when a change in hardware descriptors is detected. 

Existing performance models have been designed for different quantum hardware~\cite{Lubinski2023}, targeting specific applications,  without considering applications' performance with relation to values of typical hardware descriptors. 

\subsection{Optimization of Hybrid Workflows}
Optimization of hybrid workflows embraces different phases of workflow execution, depending on the type of quantum tasks: for circuits, optimization focuses on adaptation of circuits to the underlying qubit topology, or to perform gate-level optimization on the circuit (i.e., removing redundant parts). Also, on the classic side, data encoding has to be optimized for the underlying classic architecture. Concerning variational quantum algorithms, as shown by~\cite{cerezo2021variational,Cranganore2022}, optimization requires setting different hyperparameters, both on classic (i.e., optimizer parameters, cost function) and on the quantum side (i.e., circuit structure). Optimizations should be applied at the time of workflow execution.

Since data encoding methods require different algebraic operations on input data~\cite{10.5555/3309066}, optimization could include typical operations used in HPC to improve mathematical computations, or even the use of specific hardware (i.e., GPUs, FPGAs) that are known to perform well for these operations. For optimization of circuits according to the underlying hardware, deep-learning-based approaches~\cite{fösel2021quantum} or structural optimization approaches~\cite{bae2020quantum,Ostaszewski2021structure} can be considered. Finally, for variational quantum algorithms, hyperparameter optimization approaches can be applied, as discussed by~\cite{Cranganore2022}.

\subsection{Error Mitigation}
Current NISQ architectures are subject to noise, due to environmental factors and technological limitations~\cite{Preskill2018quantumcomputingin}. Considering the effect of noise on the output of quantum algorithms, the typical approach in hybrid systems is to perform error mitigation on classic hardware. Different approach are available, either ML-based~\cite{Domingo2023,Locher2023quantumerror} or based on error correction codes~\cite{sivak2023real}. In both cases, designing a model for error correction requires deep knowledge of the target hardware, and collecting different error metrics. Once error metrics have been identified, since error is constantly varying over time, data about such metrics have to be constantly collected in order to timely update the error model. Finally, the model should be updated timely, as soon as error goes above a given threshold.

To address these challenges, first of all, different methods to constantly monitor quantum hardware-induced errors and readout error mitigation strategies~\cite{electronics11192983} must be implemented. Moreover readout Error mitigation strategies for quantum workflows ha Then, techniques based on reinforcement learning, such as described in~\cite{schulman2017proximal,TFAgents}. These methods can be also improved by applying FPGAs~\cite{sivak2023real} or Edge AI methods, such as~\cite{demaio2022roadmap}. Updates of the model could be triggered by staleness control methods, similar to~\cite{aral2020}. 

\subsection{Integration within WMSs}
The final challenge is to integrate quantum machines into the HPC continuum. At the moment, different frameworks are available to manage quantum hardware, (i.e., Qiskit, BraKet, Pennylane). In order to fully exploit available quantum machines, a WMS should be capable of communicating with different frameworks, integrating all functionalities that are required for the execution of hybrid workflows (i.e., data encoding, transpilation, error mitigation). Each of these functionalities should be integrated with the HPC infrastructure without affecting performance of workflow execution. 

\section{Conclusion and Outlook}
\label{sec:outlook}
In this work, we describe the main components that would be needed for the execution of scientific workflows in hybrid ecosystems. Our investigation starts from the study of an MD use case, where we describe a pipeline that starts with the identification of the quantum candidates and the equivalent quantum tasks. Based on this analysis, we describe a possible architecture for a hybrid WMS, identifying software components that would facilitate the integration of classic and quantum architectures. Finally, we identify the challenges that need to be addressed for the execution of hybrid workflows. 

In the future, we plan to further investigate the integration of quantum architectures in HPC, considering different types of architectures (i.e., photonic, annealers) and different programming models (i.e., measurement-based instead of gate-based). Also, we plan to extend this study on different scientific use cases, including the development of performant software stacks.

\section*{Acknowledgements}
This work is partially funded through the projects \textit{"Runtime Control in Multi Clouds" (Rucon)}, Austrian Science Fund (FWF): Y904-N31 START-Programm 2015; Standalone Project \textit{"Transprecise Edge Computing" (Triton)}, Austrian Science Fund (FWF): P 36870-N; \textit{Trustworthy and Sustainable Code Offloading (Themis)}, Austrian Science Fund (FWF): PAT1668223; and by the Flagship Project \textit{"High-Performance Integrated Quantum Computing" (HPQC)} \# 897481 Austrian Research Promotion Agency (FFG) funded by the European Union – NextGenerationEU.
Ewa Deelman’s work was funded by the U.S. National Science Foundation under award \# 2331153.

\bibliographystyle{elsarticle-num-names} 
\bibliography{quantum-wf}

\newpage 

\appendix
\renewcommand{\theequation}{A.\arabic{equation}}
\setcounter{equation}{0}
\section{Mathematical prerequisites for quantum computing}
\begin{definition}[Hermitian Matrices/Operators]: Let $A \in M_{n ,n} (\mathbb{C})$, i.e., $A$ is a square matrix with complex entries $a_{ij}$. The matrix $A$ is said to be Hermitian (self-adjoint) if $A$ is invertible and the matrix elements satisfy the condition, 
\begin{align*}
    a_{ij} = \overline{a_{ji}}     
\end{align*}
Where $\overline{a}$ denotes its complex conjugate. Hence, a Hermitian matrix is equivalent to its transpose complex conjugate (also represented in quantum computing as a $\dagger$ subscript). A succinct matrix representation used in quantum computing is, $A = A^{\dagger} := \overline{A}^{T}$.
\end{definition}

\begin{definition}[Unitary Matrices/Operators]: Let $U \in M_{n ,n} (\mathbb{C})$. A Unitary matrix $U$ satisfies the following condition, 
\begin{align*}
    U^{-1} = U^{\dagger} := \overline{U}^{T}
\end{align*}
Hence, the inverse of a Unitary matrix is its transposed complex conjugate (also, known as \textit{Hermitian conjugate})
\end{definition}

\textbf{Tensor product:} Let A be an $m \times n$ matrix and B be a $p \times q$ matrix, 
\begin{align*}
    A = \begin{bmatrix} 
    a_{11} & a_{12} & \dots a_{1n} \\
    \vdots & \ddots &  \vdots\\
    a_{m1} &  \ldots & a_{mn} 
    \end{bmatrix} ,  
    B = \begin{bmatrix} 
    b_{11} & b_{12} & \dots b_{1q} \\
    \vdots & \ddots &  \vdots\\
    b_{p1} &  \ldots & b_{pq}
    \end{bmatrix}.
\end{align*}
The tensor product C of matrix A and B is an $mp \times nq$ dimensional matrix of the form, 
\begin{align*}
    C = \begin{bmatrix} 
    a_{11} B & a_{12} B & \dots a_{1n} B \\
    \vdots & \ddots &  \vdots\\
    a_{m1} B &  \ldots & a_{mn} B
    \end{bmatrix} \equiv  \begin{bmatrix} 
    a_{11} b_{11} & a_{11} b_{12} \ldots a_{12} b_{11}  \dots a_{1n} b_{1q} \\
    \vdots & \ddots  \vdots\\
    a_{m1} b_{11} &  a_{m1} b_{12} \dots a_{m2} b_{11}  \ldots  a_{mn} b_{pq}
    \end{bmatrix}.
\end{align*}
\medskip 

\textbf{Qubit tensor product state}:
The total degrees of freedom of a qubit composite is given by the \texttt{Tensor product} between qubit states. 
\medskip 
Consider a set of degrees of freedom associated with an $n$ dimensional Hilbert space, 
\begin{align*}
    \mathbb{H}_1 = span\big\{\ket{0}, \ket{1}, ...., \ket{n-1} \big\}, 
\end{align*}
and another set of degrees of freedom associated with an $m$ dimensional Hilbert space, 
\begin{align*}
    \mathbb{H}_2 = span\big\{\ket{0}, \ket{1}, ...., \ket{m-1} \big\}.
\end{align*}
Thus, all possible superposition states of these two qubit-conjoined Hilbert space is given 
by the tensor product, 
\begin{align}
    \mathbb{H} = \mathbb{H}_1 \otimes \mathbb{H}_2      
\end{align}

This tensor-product space is spanned by orthonormal basis vectors, 
\begin{align*}
    \{\ket{j} \otimes \ket{k} := \ket{j, k} : j = 0, 1,...., n-1; k = 0, 1,....,m-1\}
\end{align*}

Thus, an arbitrary state vector $\ket{\Phi}$ in this composite Hilbert space $\mathcal{H}_1 \otimes \mathcal{H}_2$ can be expanded in its computational basis (superposition state) reads, 

\begin{align}
    \ket{\Phi} = \sum_{j=0}^{n-1} \sum_{k=0}^{m-1} \gamma_{j,k} \ket{j} \otimes \ket{k},
\end{align}
where, the coefficients (amplitudes) $\gamma_{j,k} \in \mathbb{C}$.
Using the tensor (Kronecker) product machinery, one can now formally define a quantum register. The Hilbert space of an $n$-qubit initialized quantum register $\mathbb{H}^{\otimes n}$ is the \texttt{n-fold} tensor product state of a single qubit Hilbert space $\mathbb{H} = \mathbb{C}^2$, i.e.,  
\begin{align}
(\mathbb{C}^{2})^{\otimes n} := \underbrace{\mathbb{C}^2 \otimes ..... \otimes \mathbb{C}^2}_{\text{$n$-copies}}
\end{align}

with the basis span of Eq.(\ref{Eq.tensor}) being, 
\begin{equation*}
    \mathcal{B} := \big\{\ket{i_0} \otimes .... \otimes \ket{i_n-1}:=\ket{i_0, ...,i_{N-1}} : i_0,...,i_{n-1} \in \{0,1\} \big\}
\end{equation*}
The tensor product basis can be regarded as column vectors for e.g.,  $\ket{0} \otimes .... \otimes \ket{0} := \ket{0,0,...,0} = \begin{bmatrix} 1 ,& 0, & .....,& 0 \end{bmatrix}^{T}$, ..., $\ket{1} \otimes .... \otimes \ket{1} := \ket{1,1,...,1} = \begin{bmatrix} 0 ,& 0, & .....,& 1\end{bmatrix}^{T}$ 
\medskip

\begin{theorem}
\textbf{ Schmidt decomposition: } Let \{$\mathbb{H}_{1} \mathbb{H}_{2}, ..., \mathbb{H}_{n}$\} be Hilbert spaces of dimensions $p_1, p_2, ..., p_n$ respectively. Assume that  $p_n \geq p_{n-1} \geq.... \geq p_1$. For any state in this composite (multi-partite) system, i.e, $\ket{\xi}\in \mathbb{H}_{1} \otimes \mathbb{H}_{2} .... \otimes \mathbb{H}_{n}$, there exist orthonormal states $\{\ket{\phi_1},\ldots ,\ket{\phi_{p_1}}\} \subset \mathbb{H}_{1}$,  $\ldots, $ $\{\ket{\psi_1},\ldots ,\ket{\psi_{p_n}}\} \subset \mathbb{H}_{n}$ such that for real non-negative scalar coefficients $\lambda_i \in \mathbb{R}$, 
\begin{equation*}
\ket{\xi} = \sum_{i=1}^{r} \lambda_{i}\ket{\phi}_{i} \otimes ..... \otimes \ket{\psi}_{i}.
\end{equation*}

Where, \(\lambda_i\) are the Schmidt coefficients, and \(|\phi_l\rangle\), ...., \(|\psi_l\rangle\) are the corresponding entangled states of the composite (multi-partite) system. The number of non-zero Schmidt coefficients $r$ determines the degree of entanglement between the multiple subsystems.
\end{theorem}

\section{Quantum Programming Model}
\renewcommand{\theequation}{B.\arabic{equation}}
\setcounter{equation}{0}
\subsection{Quantum information processing and quantum computation}
Development of quantum software stacks, quantum compilers, and development in the domains of \emph{quantum software engineering} requires at least a basic understanding of quantum physics and quantum information processing. This section is an introduction to the \emph{ programming model} for quantum computation, which is especially intended for computer scientists, computer software engineers, and computational scientists who intend to migrate parts of their classical software onto quantum devices. 

\subsection{Qubit}
The basic unit of information, reflecting an on and off state of classical computers are the classical bits which can take the values $0$ and $1$ only. 
The basic building blocks of quantum computing are known as the quantum bits or in short \texttt{qubits}. 

In the case of a single qubit the associated Hilbert space is $\mathbb{H} = \mathbb{C}^2$. Thus, a qubit is a linear combination of orthonormal (vectors) basis states (superposition), denoted in the \texttt{Dirac} notation as $\ket{0}$ and $\ket{1}$. Thus the state of a qubit is a vector in a two-dimensional complex vector space~\cite{nielsen_chuang_2010}. Thus, $\ket{\Psi}$ can be expanded in the orthonormal basis states as, 

\begin{equation}
    \ket{\psi} = \alpha \ket{0} + \beta \ket{1},
\end{equation}
where, 
\begin{equation}
  \ket{0} = \begin{bmatrix}
     1 \\
    0 \\

  \end{bmatrix}
  \quad 
  \textrm{and}
  \quad
  \ket{1} = \begin{bmatrix}
    0 \\
    1 \\
  \end{bmatrix} 
\end{equation}
$\alpha, \beta \in \mathbb{C}$. These complex coefficients are also called probability amplitudes. Since probabilities must add to $1$ (normalization condition), the inner-product of the state $\ket{\psi}$ with itself (norm squared of the vector) must be equal to $1$. Hence, it is equivalent to the condition that 
\begin{align}
    \bra{\psi}\ket{\psi} := ||\psi||^2 = 1,
\end{align}
where the state $\bra{\psi} \in \mathbb{H}^{*}$ belongs to the dual vector space. This is equivalent to saying that the complex coefficients satisfy a unit norm condition, i.e., 
\begin{equation}
    |\alpha|^2 + |\beta|^2 = 1.
\end{equation}
   
An actual measurement process determines the value of the qubit, which is obtained via the amplitude probability squared. The measurement procedure yields a \textit{classical} result either corresponding to a \emph{0} or \emph{1} bit value. Thus, the measurement of  $\ket{0}$ is obtained by squaring the probability amplitude of its complex coefficient, i.e. $|\alpha|^2$, whereas $\ket{1}$ is measured by computing $|\beta|^2$. From Eq.(2) it is clear that the state vector $\ket{\psi}$ is constrained on a unit sphere $\vectorbold{S}^2 = \{\vectorbold{x} \in \mathbb{R}^3 : ||\vectorbold{x}|| = 1 \}$. Thus the geometric representation of a qubit is the so-called \texttt{Bloch-sphere} representation (The Bloch sphere was generated using the \texttt{QuTiP}~\cite{Johansson2012-ty} package). 
\begin{figure}[H]
    \centering
    \includegraphics[width=0.40\textwidth]{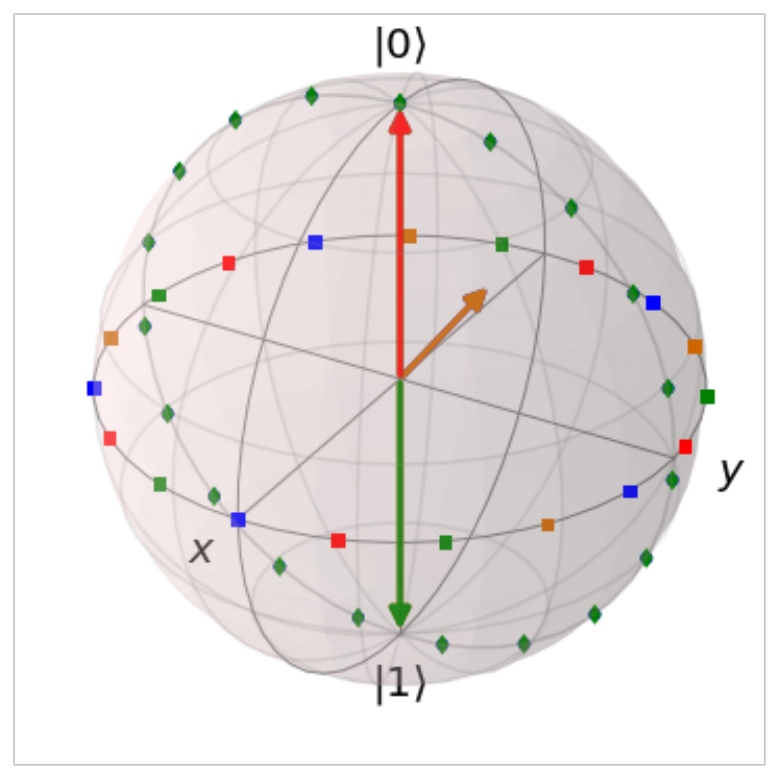}
    \caption{Bloch sphere representation of the qubit. The red-colored vector (arrow) is the state $\ket{\vectorbold{0}}$, while the green-colored vector is state $\ket{\vectorbold{1}}$. The orange vector indicates an intermediate state and the colored points correspond to the qubit rotations (infinitely many possibilities to rotate the vector) on the Bloch sphere.}
    \label{fig:bloch}
\end{figure}
\medskip

In the Bloch sphere representation, see Fig.\ref{fig:bloch}, the state vector $\ket{\psi}$ can be expressed in terms of the spherical polar coordinate basis reads,
\begin{equation}
    \ket{\psi} = cos\frac{\vartheta}{2} \ket{0} + e^{i \varphi} sin\frac{\vartheta}{2} \ket{1},
\end{equation}
where, $0 \leq \vartheta \leq \pi$ and $0 \leq \varphi \leq 2 \pi$. Thus, the state (Bloch) vector $\ket{\psi}$ can assume any of the infinitely many orientations on the Bloch-sphere. This signifies the tremendous power of information processing using quantum, since, the qubit on the Bloch sphere can exist as an infinite coherent superposition of all the states on the unit-sphere simultaneously!  This is in stark contrast to classical information processing wherein, the classical bits can only assume either of the two Boolean values $\texttt{0}$ or $\texttt{1}$ at a particular time instant~\cite{nielsen_chuang_2010}.   
\medskip 

\subsection{Quantum registers}\label{quantum-registers}

In classical computers, multiple bits are combined to form  a (classical)
register. In the same way, a sequence of $n$ initialized qubits cascaded together forms the storage-device, and is called the \textit{quantum register} or a \textit{qubit register}. Thus, an arbitrary state vector $\ket{\Psi}$ of the composite n-qubit \emph{quantum register} is a tensor product state living in a very huge Hilbert space $(\mathbb{C}^{2})^{\otimes n}$. It allows an expansion in its computational basis (orthonormal) states as, 
\begin{equation}\label{Eq.tensor}
    \ket{\Psi} = \alpha_{0,0..,0} \ket{0} \otimes ... \otimes \ket{0} + \alpha_{0,0..,1}\ket{0} \otimes ... \otimes \ket{1} + \alpha_{1,1..,1} \ket{1} \otimes ... \otimes \ket{1}.
\end{equation}

Where, the symbol $\otimes$ corresponds to the mathematical operation of a \emph{Tensor} product (see Appendix A). From here on,  we will use a more compact notation for $\ket{i_0} \otimes \ket{i_1} \otimes ... \otimes \ket{i_{n-1}} = \ket{i_0, i_1, ..., i_{n-1}}.$ 
\medskip 

Since a single qubit Hilbert (state) space is $\mathbb{H} = \mathbb{C}^2$, correspondingly, an $n$-qubit quantum register corresponds to an n-fold tensor product of the single qubit state space, i.e., 
\begin{align}
(\mathbb{C}^{2})^{\otimes n} := \underbrace{\mathbb{C}^2 \otimes ..... \otimes \mathbb{C}^2}_{\text{$n$-copies}}
\end{align}

Introducing a more succinct notation, the n-qubit quantum register, Eq.(\ref{Eq.tensor}) can be written as,
\begin{equation}
    \ket{\Psi} = \sum_{i=0}^{2^n - 1}\alpha_i \ket{i} 
\end{equation}

where, $i_0, i_1, ..., i_{n-1}$ in $\ket{i_0,...., i_{n-1}}$ represent the binary notation of $i$. The complex coefficients, $\alpha_{i_0,....,i_{n-1}} \equiv \alpha_{i} \in \mathbb{C}$. 
The normalization condition, 
\begin{equation*}
    \sum_{i=0}^{2^n-1} |\alpha_{i}|^2 = 1,
\end{equation*}

\subsection{Quantum Entanglement and Quantum Parallelism}

This subsection delves into two pivotal phenomena of quantum physics, namely \emph{quantum entanglement}, or just \emph{entanglement} and \emph{quantum parallelism}, which play fundamental roles in the fields of quantum computation quantum information processing, and quantum communication protocols~\cite{PhysRevLett.110.030501}.

\begin{itemize}
    \item Quantum entanglement: In quantum computers, a pair or multiple qubits can be correlated with each other over long distances owing to the property of \emph{quantum entanglement}. In an entangled quantum system, the state of one qubit cannot be described independently of the state of the other irrespective of their spatial separation. Thus, these states are not \emph{separable} (no tensor product representation) and cannot be written as a \emph{tensor product} state. The degree of entanglement between the composite systems can be quantified using the \emph{Schmidt decomposition} (see Appendix A for definition). Manipulating the state of one of the qubits instantaneously changes the state of the other one in a predictable way.  Thus, quantum entanglement enables the users to create complex quantum circuits (cf. Appendix C), wherein an operation on a single qubit instantaneously affects the state of another qubit that it is correlated with. Entanglement enables quantum computers to perform parallel computations by adding additional qubits resulting in an exponential increase in its number-crunching capabilities. 

\item Quantum parallelism: Entanglement also influences the superposition of multiple qubits by allowing them to be in a joint superposition state, leading to another intrinsic feature of quantum processors, namely, \emph{quantum parallelism}~\citep{Bub2007,Djordjevic2022}. 
It is evident from Eq.(\ref{Eq.tensor}) that a quantum register with $n$-qubits is in a quantum superposition of $2^n$ states, ie., \emph{all classical alternatives at once}. Hence, one can simultaneously manipulate all the $2^n$ possibilities that exist in this \emph{huge} extended vector space.

\end{itemize}

\subsection{Quantum Logic Gates, Quantum Circuit Model and Quantum Algorithms}

Quantum (logic) gate architectures are the quantum analogs of classical logic gates. A quantum gate $\hat{U}$ is represented as a unitary operator (matrix) is a generalization of rotation on complex vector spaces. An \texttt{n}-qubit quantum register  The unitary operator transformations are inner-product (or norm) preserving. By definition, these unitary operator inverse is  its own transposed complex conjugate, also called \textit{Hermitian conjugate} (see Appendix A), i.e., 
\begin{equation*}
    \hat{U} \hat{U}^{\dagger} = \hat{U}^{\dagger} \hat{U} = \mathds{1}.
\end{equation*}
The $U^{\dagger}$ quantum gate reverses the computation by undoing the gate operator. Thus, in quantum devices, there exists another unique feature of reversible computation using the inverse Unitary operators that is not possible in classical architectures.

\subsubsection{Quantum Circuit Model:}  The initialized quantum register together with the qubit operation producing single or multi-qubit quantum gate sequences form the basic building blocks of a \texttt{quantum circuit}. There are multiple quantum gates that one uses to perform computations on the state vector. Some of the most frequently used single qubit gates in quantum information processing and computation are the so-called, \texttt{PAULI} gates (operators) \{$I$, $X$, $Y$, $Z$\}, the \texttt{HADAMARD} gate $H$ and the $P$-gate (see Appendix B).

A \emph{quantum algorithm} is a collection of unitary quantum gates that are assembled to successively perform one or many \emph{unitary transformations} (computations) on a quantum register, in order to achieve a specific computational task. In short, these perform targeted rotations (operations) on a single qubit or a quantum register by mapping it onto another state on the Bloch-sphere. 
Thus quantum gates perform linear unitary transformations (manipulation) on an input quantum register $\ket{R}$ and map them onto an output quantum register  $\ket{Q}$ as follows, 
\begin{equation}
    \hat{U}\ket{R} = \ket{Q} := \sum_{i=1}^{N} \alpha_{i_0,...,i_{N-1}} \ket{i_0,....,i_{N-1}} 
\end{equation}
The result of the quantum algorithm $U$ is obtained by measuring the quantum register $\ket{Q}$ with a probability $|\alpha_{i_0, .., i_{N-1}}|^2$. Due to the probabilistic nature of the \emph{quantum processing units} (QPUs), over different executions of $U$ followed by a measurement to determine the result yields different bit-strings according to their probabilities. This is to say that a single execution of a quantum algorithm is like performing a random experiment. Thus, an algorithm is typically executed multiple times, producing a probability density function (probability distribution) of results rather than a single value.  The most probable result in this statistical sample space corresponds to the actual \emph{result} of the quantum algorithm.

\section{Unitary Quantum Logic Gates}\label{quantum-gates}
\renewcommand{\theequation}{C.\arabic{equation}}
\setcounter{equation}{0}
\paragraph{Single Qubit Gate operations}
\begin{GrayBox}
    \textbf{X} gate: The \texttt{Pauli-X} gate is the quantum analog of the classical \texttt{NOT} gate. It performs a bit flip (NOT) operation $\ket{x} \mapsto \ket{\neg x}$.
\end{GrayBox}
The matrix representation of the X-gate reads, 
\begin{align*}
X = 
\left[\begin{array}{cc}
     0 & 1 \\ 
     1 & 0 \\ 
\end{array}\right],
\end{align*}
and the quantum circuit and truth table is, 
\begin{figure}[!htb]
    \centering
    \begin{minipage}{0.5\columnwidth}
     \begin{quantikz}[transparent]
    \lstick{$\ket{x}$} & \gate[style={fill=red!20}]{X} & \rstick{$\ket{\neg  x}$}\qw 
\end{quantikz}
    \end{minipage}%
    \begin{minipage}{.5\columnwidth}
              \begin{tabular}{||c||c||}
    \hline
    Input Qbits & Output Qbits\\
    \hline
    \hline
    $\ket{0}$ &  $\ket{1}$ \\
    $\ket{1}$ &  $\ket{0}$ \\
    \hline
    \end{tabular}
    \end{minipage}
\end{figure}

\begin{GrayBox}
    \textbf{Z} gate: The \texttt{Pauli-Z} gate is the sign \texttt{Flip} gate following the operation, $\ket{x} \mapsto (-1)^x \ket{x}$.
\end{GrayBox}
The matrix representation of the Z-gate reads, 
\begin{align*}
Z = 
\left[\begin{array}{cc}
     1 & 0 \\ 
     0 & -1 \\ 
\end{array}\right],
\end{align*}
with the quantum circuit and truth table being, 
\begin{figure}[!htb]
    \centering
    \begin{minipage}{0.5\columnwidth}
     \begin{quantikz}[transparent]
    \lstick{$\ket{x}$} & \gate[style={fill=red!20},label style=blue]{Z} & \rstick{$(-1)^x \ket{x}$}\qw 
\end{quantikz}
    \end{minipage}%
    \begin{minipage}{.5\columnwidth}
              \begin{tabular}{||c||c||}
    \hline
    Input Qbits & Output Qbits\\
    \hline
    \hline
    $\ket{0}$ &  $\ket{0}$ \\
    $\ket{1}$ &  -$\ket{1}$ \\
    \hline
    \end{tabular}
    \end{minipage}
\end{figure}
\begin{GrayBox}
    \textbf{Y} gate: The Pauli-\texttt{Y} gate performs a rotation by $\pi$ 
around the y-axis. 
\end{GrayBox}
The matrix representation of the Y-gate reads, 
\begin{align*}
Y = 
\left[\begin{array}{cc}
     0 & -i \\ 
     i & 0 \\ 
\end{array}\right],
\end{align*}
with the quantum circuit and truth table being, 
\begin{figure}[!htb]
    \centering
    \begin{minipage}{0.5\columnwidth}
     \begin{quantikz}[transparent]
    \lstick{$\ket{x}$} & \gate[style={fill=red!20},label style=blue]{Y} & \rstick{$i (-1)^x \ket{x}$}\qw 
\end{quantikz}
    \end{minipage}%
    \begin{minipage}{.5\columnwidth}
              \begin{tabular}{||c||c||}
    \hline
    Input Qbits & Output Qbits\\
    \hline
    \hline
    $\ket{0}$ &  i$\ket{0}$ \\
    $\ket{1}$ &  -i$\ket{1}$ \\
    \hline
    \end{tabular}
    \end{minipage}
\end{figure}
\begin{GrayBox}
${\bm{I}}$ gate: The Identity gate $I$ leaves all states unchanged, 
\end{GrayBox}
The matrix representation of the Identity gate reads, 
\begin{align*}
I = 
\left[\begin{array}{cc}
     1 & 0 \\ 
     0 & 1 \\ 
\end{array}\right],
\end{align*}
with the quantum circuit and truth table being, 
\begin{figure}[!htb]
    \centering
    \begin{minipage}{0.5\columnwidth}
     \begin{quantikz}[transparent]
    \lstick{$\ket{x}$} & \gate[style={fill=red!20}]{I} & \rstick{$\ket{x}$}\qw 
\end{quantikz}
    \end{minipage}%
    \begin{minipage}{.5\columnwidth}
              \begin{tabular}{||c||c||}
    \hline
    Input Qbits & Output Qbits\\
    \hline
    \hline
    $\ket{0}$ &  $\ket{0}$ \\
    $\ket{1}$ &  $\ket{1}$ \\
    \hline
    \end{tabular}
    \end{minipage}
\end{figure}
\medskip 

One of the most important and fundamental  single qubit gate is the \texttt{HADAMARD} gate. The Hadamard gate prepares a superposition state, i.e., 
\begin{align}
H \ket{0} &= \frac{1}{\sqrt{2}} 
\left[\begin{array}{cc}
     1 & 1 \\ 
     1 & -1 \\ 
\end{array}\right]
 \left[\begin{array}{c}
     1 \\ 
     0 \\ 
\end{array}\right] = 
 \frac{1}{\sqrt{2}} \left[\begin{array}{c}
     1 \\ 
     1 \\ 
\end{array}\right] = \frac{\ket{0} + \ket{1}}{\sqrt{2}} := \ket{+}
\end{align}
\begin{align}
H \ket{1} &= \frac{1}{\sqrt{2}} 
\left[\begin{array}{cc}
     1 & 1 \\ 
     1 & -1 \\ 
\end{array}\right]
 \left[\begin{array}{c}
     0 \\ 
     1 \\ 
\end{array}\right] = 
 \frac{1}{\sqrt{2}} \left[\begin{array}{c}
     1 \\ 
     -1 \\ 
\end{array}\right] = \frac{\ket{0} - \ket{1}}{\sqrt{2}} := \ket{-}
\end{align}

\texttt{P} gate: The phase (shift) gate \texttt{P} are single qubit gates that leaves the basis state $\ket{0}$ unchanged while induces a phase on state $\ket{1}$, i.e., 
\begin{align}
P(\varphi) \ket{0} &= 
\left[\begin{array}{cc}
     1 & 0 \\ 
     0 & e^{i \varphi} \\ 
\end{array}\right]
 \left[\begin{array}{c}
     1 \\ 
     0 \\ 
\end{array}\right] = 
\left[\begin{array}{c}
     1 \\ 
     0 \\ 
\end{array}\right] = \ket{0}
\end{align}

\begin{align}
P(\varphi) \ket{1} &= 
\left[\begin{array}{cc}
     1 & 0 \\ 
     0 & e^{i \varphi} \\ 
\end{array}\right]
 \left[\begin{array}{c}
     0 \\ 
     1 \\ 
\end{array}\right] = 
e^{i \varphi}
\left[\begin{array}{c}
     0 \\ 
     1 \\ 
\end{array}\right] = e^{i \varphi} \ket{1}
\end{align}

\paragraph{Two Qubit Gate operations}
\begin{GrayBox}
\textbf{CNOT gate}: The \texttt{CNOT} or the   \texttt{Controlled-X} gate  is a 2-qubit gate which flips the second qubit (target qubit) if and only if the first qubit (control qubit) is state $\ket{1}$. It is the quantum analogue of the classical \texttt{XOR} gate and maps the state $\ket{x, y} \mapsto \ket{x \oplus y}$, where the symbol $\oplus$ denotes the XOR logic operation. 
\end{GrayBox}
The matrix representation of the controlled-X gate reads, 
\begin{align*}
\ket{0}\bra{0} \mathbb{I} + \ket{1}\bra{1} X &= 
\left[\begin{array}{cccc}
     1 & 0 & 0 & 0 \\ 
     0 & 1 & 0 & 0 \\
     0 & 0 & 0 & 1 \\
     0 & 0 & 1 & 0 \\
\end{array}\right]
\end{align*}

The quantum circuit and the corresponding truth (logic) table is shown below, 


\begin{figure}[!htb]
    \centering
    \begin{minipage}{0.5\columnwidth}
     \begin{quantikz}[transparent]
    \lstick{$\ket{x}$} & \qw\gategroup[2,steps=3,style={dashed,
    rounded corners,fill=blue!20, inner xsep=2pt},
    background,label style={label position=below,anchor=
    north,yshift=-0.2cm}]{{\sc \texttt{CNOT}}}  & \ctrl{1} & \qw & \rstick{$\ket{x}$}\qw \\
    \lstick{$\ket{y}$} & \qw & \targ{} & \qw & \rstick{$\ket{x \oplus y}$}\qw 
\end{quantikz}
    \end{minipage}%
    \begin{minipage}{.5\columnwidth}
              \begin{tabular}{||c||c||}
    \hline
    Input Qbits & Output Qbits\\
    \hline
    \hline
    $\ket{00}$ &  $\ket{00}$ \\
    $\ket{01}$ &  $\ket{01}$ \\
    $\ket{10}$ &  $\ket{11}$ \\
    $\ket{11}$ &  $\ket{10}$ \\
    \hline
    \end{tabular}
    \end{minipage}
\end{figure}

\begin{GrayBox}
\textbf{CZ gate}: The \texttt{controlled Z} gate is a 2-qubit gate flips the sign of the state $\ket{11}$, while leaving the other states unaffected, ie., $\ket{x, y} \mapsto (-1)^x\ket{x ,y}$
\end{GrayBox}

The matrix representation of the controlled-Z reads, 
\begin{align*}
\left[\begin{array}{cccc}
     1 & 0 & 0 & 0 \\ 
     0 & 1 & 0 & 0 \\
     0 & 0 & 1 & 0 \\
     0 & 0 & 0 & -1 \\
\end{array}\right],
\end{align*}
and the corresponding quantum circuit and truth table are also presented below,  
\begin{figure}[!htb]
    \centering
    \begin{minipage}{0.5\columnwidth}
     \begin{quantikz}[transparent]
    \lstick{$\ket{x}$} & \qw\gategroup[2,steps=3,style={dashed,
    rounded corners,fill=pink!50, inner xsep=2pt},
    background,label style={label position=below,anchor=
    north,yshift=-0.2cm}]{{\sc \texttt{CZ}}}  & \ctrl{1} & \qw & \rstick{$\ket{x}$}\qw \\
    \lstick{$\ket{y}$} & \qw & \gate{Z} & \qw & \rstick{$(-1)^x\ket{y}$}\qw 
\end{quantikz}
    \end{minipage}%
    \begin{minipage}{.5\columnwidth}
  \begin{tabular}{||c||c||}
\hline
Input Qbits & Output Qbits\\
\hline
\hline
$\ket{00}$ &  $\ket{00}$ \\
$\ket{01}$ &  $\ket{01}$ \\
$\ket{10}$ &  $\ket{10}$ \\
$\ket{11}$ &  -$\ket{11}$ \\
\hline
\end{tabular}
    \end{minipage}
\end{figure}

\begin{GrayBox}
\textbf{SWAP gate}: The 2-qubit SWAP gate swaps the qubit states and maps a state $\ket{a, b} \mapsto \ket{b, a}$.
\end{GrayBox}
The permutation matrix representation reads, 
\begin{align}
\left[\begin{array}{cccc}
     1 & 0 & 0 & 0 \\ 
     0 & 0 & 1 & 0 \\
     0 & 1 & 0 & 0 \\
     0 & 0 &  0 & 1 \\
\end{array}\right]
\end{align}
while the quantum circuit representation and truth table is, 
\begin{figure}[H]
    \centering
    \begin{minipage}{0.5\columnwidth}
     \begin{quantikz}[transparent]
    \lstick{$\ket{x}$} & \qw\gategroup[2,steps=3,style={dashed,
    rounded corners,fill=purple!60, inner xsep=2pt},
    background,label style={label position=below,anchor=
    north,yshift=-0.2cm}]{{\sc \texttt{SWAP}}}  & \swap{1} & \qw & \rstick{$\ket{y}$}\qw \\
    \lstick{$\ket{y}$} & \qw & \swap{-1} & \qw & \rstick{$\ket{x}$}\qw 
\end{quantikz}
    \end{minipage}%
    \begin{minipage}{.5\columnwidth}
  \begin{tabular}{||c||c||}
\hline
Input Qbits & Output Qbits\\
\hline
\hline
$\ket{00}$ &  $\ket{00}$ \\
$\ket{01}$ &  $\ket{10}$ \\
$\ket{10}$ &  $\ket{01}$ \\
$\ket{11}$ &  $\ket{11}$ \\
\hline
\end{tabular}
    \end{minipage}
\end{figure}

\emph{Pauli Group}: 
Pauli matrices generate a discrete group closed under multiplication, called the Pauli group $P_n$.
The set $P_n$ consists of \texttt{n}-fold tensor product of the Pauli operators (Pauli strings) multiplied by a factor $\gamma \in \{\pm{1}, \pm{i}\}$ accounting to 16-elements. An example of the Pauli group for $n = 2$, are the 2-fold tensor product of the Pauli gates, $\{\gamma I \otimes I, \gamma I \otimes X, \gamma I \otimes Y, \gamma I \otimes Z, \gamma X \otimes I, \gamma X \otimes X, .... \gamma Z \otimes Z\}$. 
\medskip 

\begin{definition}
    The \emph{normalizer} of a subgroup H of a group (or semigroup) G is defined as: 
\end{definition}
\begin{equation*}
        N_G(H) = \big\{g \in G | gHg^{-1} = H\big\}
    \end{equation*}

\emph{Clifford gates}: The Clifford group on n qubits, $\mathcal{C}_n$, are the set of unitary operations that normalize the Pauli group $P_n$. That is, $U \in \mathcal{C}_n$ if $UpU^{\dagger} \in P_n,  \ \forall p \in P_n$. 
The \emph{Clifford gates} are unitary operators in $\bigcup_{n\geq 1} \mathcal{C}_n$. A quantum circuit constructed merely out of Clifford gates is called the Clifford circuit~\cite{Grier2022classificationof}. It performs qubit operations on some designated set of initialized qubits, while preserving the state of remaining the ancilla qubits.

The Clifford gate set consists of three gates, namely, the CNOT (controlled-NOT), the Hadamard gate $H$ and the Phase gate $P$.



\subsection{Three and Multi-qubit gate operations}
\begin{GrayBox}
\textbf{TOFFOLI gate}: The \texttt{CCNOT} (controlled-controlled NOT gate) or the \texttt{TOFFOLI} gate is a three qubit universal reversible quantum gate. If the first two qubits are in state $\ket{1}$, then it flips the last qubit state, i.e.,  $\ket{x,y,z} \mapsto \ket{x,y, z \oplus (x \land y)}$
\end{GrayBox}
The Toffoli gate in its matrix form reads,  
\begin{align*}
\left[\begin{array}{cccccccc}
     1 & 0 & 0 & 0 & 0 & 0 & 0 & 0\\ 
     0 & 1 & 0 & 0 & 0 & 0 & 0 & 0\\
     0 & 0 & 1 & 0 & 0 & 0 & 0 & 0\\
     0 & 0 & 0 & 1 & 0 & 0 & 0 & 0\\
     0 & 0 & 0 & 0 & 1 & 0 & 0 & 0\\
     0 & 0 &  0 & 0 & 0 & 1 & 0 & 0\\
     0 & 0 & 0 & 0 & 0 & 0 & 0 & 1\\
     0 & 0 &  0 & 0 & 0 & 0 & 1 & 0\\
\end{array}\right].
\end{align*}

The three qubit Toffoli gate has an equivalent quantum circuit representation solely interms of two qubit gates $H$, $T$ and $T^{\dagger}$. and are . Also, the truth table of this quantum logic gate are presented below, 

\begin{figure}[!htb]
    \centering
     \begin{quantikz}[transparent]
     \qw &  \qw\gategroup[3,steps=3,style={dashed,
   rounded corners,fill=orange!30, inner xsep=2pt},
    background,label style={label position=below,anchor=
 north,yshift=-0.2cm}]{{\sc \texttt{CCNOT}}} & \ctrl{2} & \qw & \qw \\
    \qw & \qw & \ctrl{1} & \qw & \qw \\ 
    \qw & \qw & \targ & \qw &  \qw & \qw 
    \end{quantikz}%
    $\equiv$
    \begin{minipage}{0.80\columnwidth}
\begin{quantikz}[row sep=0.9mm,column sep=0.9mm]
& \qw      & \qw      & \qw              & \ctrl{2} & \qw      & \qw      & \qw              & \ctrl{2} & \qw      & \ctrl{1} & \gate[style={fill=yellow!50}]{T}         & \ctrl{1} & \qw \\
& \qw      & \ctrl{1} & \qw              & \qw      & \qw      & \ctrl{1} & \qw              & \qw      & \gate[style={fill=yellow!50}]{T} & \targ{}  & \gate[style={fill=yellow!100}]{T^\dagger} & \targ{}  & \qw \\
& \gate[style={fill=red!50}]{H} & \targ{}  & \gate[style={fill=yellow!100}]{T^\dagger} & \targ{}  & \gate[style={fill=yellow!50}]{T} & \targ{}  & \gate[style={fill=yellow!100}]{T^\dagger} & \targ{}  & \gate[style={fill=yellow!50}]{T} & \gate[style={fill=red!50}]{H} & \qw              & \qw      & \qw
\end{quantikz} 
    \end{minipage}
\end{figure}
\begin{center}


     \begin{tabular}{||c||c||}
\hline
Input Qbits & Output Qbits\\
\hline
\hline
$\ket{000}$ &  $\ket{000}$ \\
$\ket{001}$ &  $\ket{001}$ \\
$\ket{010}$ &  $\ket{010}$ \\
$\ket{011}$ &  $\ket{011}$ \\
$\ket{100}$ &  $\ket{100}$ \\
$\ket{101}$ &  $\ket{101}$ \\
$\ket{110}$ &  $\ket{111}$ \\
$\ket{111}$ &  $\ket{110}$ \\
\hline
\end{tabular}
\end{center}
\medskip 

\begin{GrayBox}
\textbf{FREDKIN gate}: The \texttt{CSWAP} (controlled-swap gate) is a three qubit universal reversible quantum gate. If and only if the first  qubit state is state $\ket{1}$, it leaves the the first qubit  unchanged and swaps the last two bits.
\end{GrayBox}
The Fredkin gate in its permutation matrix form reads,  
\begin{align*}
\left[\begin{array}{cccccccc}
     1 & 0 & 0 & 0 & 0 & 0 & 0 & 0\\ 
     0 & 1 & 0 & 0 & 0 & 0 & 0 & 0\\
     0 & 0 & 1 & 0 & 0 & 0 & 0 & 0\\
     0 & 0 & 0 & 1 & 0 & 0 & 0 & 0\\
     0 & 0 & 0 & 0 & 1 & 0 & 0 & 0\\
     0 & 0 &  0 & 0 & 0 & 0 & 1 & 0\\
     0 & 0 & 0 & 0 & 0 & 1 & 0 & 0\\
     0 & 0 &  0 & 0 & 0 & 0 & 0 & 1\\
\end{array}\right].
\end{align*}

The three qubit Fredkin gate has an equivalent quantum circuit that can be solely constructed interms of CNOT, $V$ \&  $V^{\dagger}$ qubit gates ~\cite{PhysRevA.53.2855}. Here, \begin{align*}
V = 
\left[\begin{array}{cc}
     0 & 1 \\ 
     1 & 0 \\ 
\end{array}\right]^{\frac{1}{2}},
\end{align*} is the square-root of the Pauli $X$ gate. Infact, the \texttt{FREDKIN} gate is just the \texttt{TOFFOLI} gate with two CNOTs on its either sides. The corresponding truth table for the \texttt{FREDKIN} gate is presented below, 

\begin{figure}[!htb]
    \centering
  \begin{quantikz}[transparent]
    \lstick{$\ket{x}$} & \qw\gategroup[3,steps=3,style={dashed,
    rounded corners,fill=yellow!40, inner xsep=2pt},
    background,label style={label position=below,anchor=
    north,yshift=-0.2cm}]{{\sc \texttt{C-SWAP}}}  & \ctrl{2} & \qw & \rstick{$\ket{x}$}\qw \\
    \lstick{$\ket{y}$} & \qw & \targX{} & \qw & \rstick{$\ket{y}$}\qw \\ 
    \lstick{$\ket{z}$} & \qw & \swap{-1} & \qw & \rstick{$\ket{z \oplus x}$}\qw 
\end{quantikz}
    $\equiv$
    \begin{minipage}{\columnwidth}

\begin{quantikz} 
& & \qw \gategroup[3,steps=9,style={dashed,rounded
    corners,fill=yellow!40, inner
    xsep=2pt},background,label style={label
    position=below,anchor=north,yshift=-0.2cm}]{{\texttt{C-SWAP}}}  & \qw & \qw &  \ctrl{2} & \ctrl{1} & \qw &  \qw & \ctrl{1} \qw & \qw & \qw & & \\ 
&  &
    \targ{} &\ctrl{1} & \qw & \qw & \targ{}  &  \ctrl{1} & \qw & \targ{} & \targ{} & \qw \\
& & \ctrl{-1} & \gate[style={fill=green!100}]{V} & \qw & \gate[style={fill=green!100}]{V} & \qw  &  \gate[style={fill=green!100}]{V^{\dagger}} & \qw & \qw & \ctrl{-1} & \qw \end{quantikz}
    \end{minipage}
\end{figure}
\begin{center}

\begin{tabular}{||c||c||}
\hline
Input Qbits & Output Qbits\\
\hline
\hline
$\ket{000}$ &  $\ket{000}$ \\
$\ket{001}$ &  $\ket{001}$ \\
$\ket{010}$ &  $\ket{010}$ \\
$\ket{011}$ &  $\ket{011}$ \\
$\ket{100}$ &  $\ket{100}$ \\
$\ket{101}$ &  $\ket{110}$ \\
$\ket{110}$ &  $\ket{101}$ \\
$\ket{111}$ &  $\ket{111}$ \\
\hline
\end{tabular}
\end{center}

\section{Engineering complex quantum circuits}\label{complex-quantum-circuits}
\renewcommand{\theequation}{D.\arabic{equation}}
\setcounter{equation}{0}
\textbf{Entangled states}: Separable quantum state can be expanded in its computational basis as, $\ket{\Psi} = \ket{\psi_0} \otimes \ket{\psi_1} \otimes .... \otimes\ket{\psi_{n-1}}$. Whereas, an \emph{Entangled} quantum state $\ket{\zeta}$ cannot be decomposed into tensor product states, i.e.,  
$\ket{\zeta} \neq \ket{\xi_0} \otimes \ket{\xi_1} \otimes .... \otimes\ket{\xi_{n-1}}$.
The most simple and maximally entangled quantum states can be acheived by entangling 2-qubits in 4 different manners, also known as the \texttt{Bell} states or \texttt{EPR} (\textit{Einstein-Podolski-Rosen}) states,  
\begin{align*}
    \ket{\Phi^+} = \frac{\ket{00} + \ket{11}}{\sqrt{2}}, 
    \ket{\Phi^-} = \frac{\ket{00} - \ket{11}}{\sqrt{2}}, \\ 
    \ket{\Psi^+} = \frac{\ket{01} + \ket{10}}{\sqrt{2}},  
    \ket{\Psi^-} = \frac{\ket{01} - \ket{10}}{\sqrt{2}}, \\ 
\end{align*}
We demonstrate the preparation of the $\ket{\Phi^{+}}$ state via a quantum circuit. 
\begin{center}
\begin{quantikz}
\gategroup[wires=2,steps=11,style={rounded corners,fill=blue!20}, background]{}
&\lstick{$|{0}\rangle$} & \gate[style={fill=red!50}]{H}&\ctrl{1} & \qw &\qw &\qw 
 \rstick[wires=2]{$\frac{|{00}\rangle + |{11}\rangle}{\sqrt{2}}$} 
\\
&\lstick{$|{0}\rangle$} & \qw& \targ{} & \qw & \qw &\qw

&&&&
\end{quantikz}
\end{center}

In the 3-qubit case, there exists non bi-separbale classes of entangled states in quantum computing are for e.g., the 3-qubit \texttt{Greenberger-Horne-Zeilinger} (GHZ) state~\cite{PhysRevLett.106.130506}, 
\begin{align*}
    \ket{\mathbf{GHZ}} = \frac{1}{\sqrt{2}}(\ket{000} + \ket{111})
\end{align*}

\begin{center}
\begin{quantikz}
\gategroup[wires=3,steps=11,style={rounded corners,fill=blue!20}, background]{}
&\lstick{$|{0}\rangle$} & \gate[style={fill=red!50}]{H}&\ctrl{1} & \ctrl{2} &\qw &\qw 
 \rstick[wires=3]{$\frac{|{000}\rangle + |{111}\rangle}{\sqrt{2}}$} 
\\
&\lstick{$|{0}\rangle$} & \qw& \targ{} & \qw & \qw &\qw
\\
&\lstick{$|{0}\rangle$} & \qw & \qw & \targ{} & \qw &\qw
&&&&
\end{quantikz}
\end{center}

another highly important entangled 3-qubit state that is inequivalent to the GHZ state is the \texttt{W} state~\cite{PhysRevA.62.062314}, 
\begin{align*}
    \ket{\mathbf{W}} = \frac{1}{\sqrt{3}} (\ket{001} + \ket{010} + \ket{100})
\end{align*}
\newpage

\section{Matrix Size}\label{matrix-size}
\renewcommand{\theequation}{E.\arabic{equation}}
\setcounter{equation}{0}
The symmetric bipartite matrix $B_{IJ}$ can be partioned into a block matrix form. 

\begin{equation}
B_{\text{IJ}} = \begin{pmatrix}
        0^{n \times n} & E_{\text{IJ}} \\
        E_{\text{IJ}}^T  & 0^{n \times n} \\
     \end{pmatrix}
     \label{eq:md-input-matrix}
\end{equation}

The diagonal entries of the bipartite block matrix contains the zero matrix. The off-diagonal entries $E_{\text{IJ}}$ ($n \times n$ matrix) and its transpose $E_{\text{IJ}}^T$ contain as entries the Euclidean distances $d_{ij}$ as defined in Eq.(\ref{Eq:dist}). The Euclidean distance (metric)  is a function defined on vector space $\mathbb{V}$,
\begin{align*}
  d : \mathbb{V} \times \mathbb{V} \mapsto \mathbb{R}.
\end{align*} 
Therefore, the block matrix $E_{\text{IJ}}$ takes values only over the field $\mathbb{R}$. The  matrix representation is given by, 
\begin{equation}
 E_{\text{IJ}} = \begin{pmatrix} 
d_{ij}^{11}  & \cdots &  d_{ij}^{1n} \\
\vdots & \ddots & \vdots \\
d_{ij}^{n1} & \cdots &  d_{ij}^{nn}   \\
\end{pmatrix}.
\end{equation}

For given two segments $I$ and $J$, the Euclidean metric between $C_\alpha$ carbon atoms $i$ and $j$ reads,
\begin{equation}
   d_{ij} =  d(i,j) = \sqrt{(i_x - j_x)^2 + (i_y - j_y)^2 + (i_z - j_z)^2}.
   \label{eq:euclidean}
\end{equation}
\begin{theorem}
Input bipartite distance matrix $B_{IJ}$ is a \textbf{Hermitian} matrix.
\end{theorem}
\begin{proof}
Let $B_{IJ} \in \mathbb{R}^{2n} \times \mathbb{R}^{2n}$, with $n \in \mathbb{N}$. Note that from Equation~\ref{eq:md-input-matrix},
\begin{equation*}
B_{\text{IJ}}^{\dagger} = \begin{pmatrix}
    0^{n \times n} & \overline{E}_{\text{IJ}} \\
        \overline{E}_{\text{IJ}}^T & 0^{n \times n} \\ 
     \end{pmatrix}^T = \begin{pmatrix}
    0^{n \times n} & E_{\text{IJ}} \\
        E_{\text{IJ}}^T & 0^{n \times n} \\ 
     \end{pmatrix}^T = B_{\text{IJ}}
\end{equation*}
Hence, $B_{\text{IJ}}$ is a \textbf{real symmetric matrix}. 
Since, every real symmetric matrix is an Hermitian matrix (operator), $B_{\text{IJ}}$ is an Hermitian operator.
\end{proof}
Thus, the set of eigenvalues $\Lambda$ of the bipartite distance matrix all belong in $\mathbb{R}$. Thus, the Hermitian bipartite matrix $\hat{B}_{\text{IJ}}$ qubit encoding (via mapping on to Pauli operators) without any further manipulation.

\section{C-SWAP test methodology for $C_{\alpha}$ atom distance estimation}
\renewcommand{\theequation}{B.\arabic{equation}}
\setcounter{equation}{0}
First introduced in the context of quantum fingerprinting~\cite{PhysRevLett.87.167902}. The procedure for engineering the SWAP test circuit involves, (i) Entangling qubit registers consisting of the mapped classical data with an ancillary/helper qubit $\ket{0}$ and, (ii) estimating the inner product between two different states through repeated measurements of the ancillary qubit. Thus, this quantum algorithm measures the so-called \emph{fidelity}, which is nothing but the inner product or overlap between two different quantum states.  The fidelity $F$ between two normalized quantum states $\ket{\phi}, \ket{\psi}$ is mathematically expressed as, $F(\phi, \psi)= |\bra{\phi}\ket{\psi}|^2 \in [0, 1]$. Hence, the higher the fidelity, the closer are the quantum states to each other, i.e, $F(\phi, \psi) = 1$, meaning the quantum states are parallel while $F(\phi, \psi) = 0$ meaning orthogonal quantum states. 

From an application point of view, this quantum routine has been integrated into larger quantum circuits to achieve the target tasks. For example, it has been used as a primary engine for speeding-up matrix multiplications boiling down to merely $\mathcal{O}(N^2)$~\cite{Zhang2016} time-complexity. Moreover it has been extensively used in the  domains of \textit{quantum machine learning} and \emph{big-data} analysis (also sometimes called \emph{quantum big data}) due to its exponential speedup offered in calculating distances between huge amounts of vector (tensor) datasets~\citep{https://doi.org/10.48550/arxiv.1307.0411, 10.5555/2871393.2871400, kopczyk2018quantum}.   

\paragraph{Circuit Description}\label{circ_des}
We initialize an ancillary (helper) qubit $\ket{0}$, two quantum registers $\ket{\psi}$ and $\ket{\phi}$.
The initial state of combined tensor product is 
\begin{equation}
\ket{\Psi_A} = \ket{0} \otimes \ket{\phi} \otimes \ket{\psi},
\end{equation}

Encoded into the states $\ket{\phi}$ and $\ket{\psi}$ are classical data atom coordinates using the \texttt{AMPLITUDE} encoding method. In \emph{amplitude encoding}, a classical vector or tensor, $X = \big[x_1, x_2, ..., x_n\big]^{T} \in \mathbb{R}^N $ is mapped onto the quantum device by implementing the following algorithm written as a pseudocode: 
\begin{algorithm}[H]
\textbf{Input:} Classical data, $X  = \big[x_0, x_1, ..., x_{p-1}\big]^{T} \in \mathbb{R}^p$, number of qubits $n$.\\
\textbf{Result:} Quantum data with coefficients of $X \in \mathbb{R}^p$ encoded as amplitudes of the state-vector $X \mapsto \ket{Q_X} = \frac{1}{||\vectorbold{x}||^2}\sum_{i=0}^{n-1} x_i \ket{i}$
\begin{algorithmic}[1]
\STATE $p \gets  \textrm{\texttt{LEN}$(X)$}  $\label{alg:genInput} \\ 

\IF {$\lceil{log_2 (n)}\rceil - p = 0$}  
    \STATE $Q_X \gets X$ \COMMENT{Calculate magnitude (norm) squared of $X$}\\ 
    \STATE $Q_X \mapsto{\ket{Q_X}} = \frac{1}{||X||^2}\sum_{i=0}^{n-1} x_i \ket{i}$
\ENDIF
\IF {$\lceil{log_2 (n)}\rceil - p =  (k - 1) $} 
    \STATE $Q_X \gets  \underbrace{\big[\underbracket{x_0, x_1, ..., x_{p-1}}_{\text{p-entries}}, \underbracket{0,.....,0}_{\text{(k-1) -entries}}\big]^{T}}_{\text{$p + (k-1) = n $}}$  \COMMENT{Pad the vector $X$ with ($k-1$) $0$'s to convert to dimension $n$.} \\
    \STATE $Q_X \mapsto{\ket{Q_X}} = \frac{1}{||X||^2}\big(\sum_{i=0}^{p-1} x_i \ket{i} + \sum_{i=p}^{n-1} 0 \ket{i}$\big) \COMMENT{Calculate magnitude (norm) squared of $X$}\\  
\ENDIF

\RETURN $\ket{Q_X}$\\
\end{algorithmic}
\caption{Amplitude Encoding Schema}
\label{alg:data-driven}
\end{algorithm}
\medskip

Thus, an $n$-dimensional classical data vector (tensor) can be efficiently encoded into the wavefunction (state vector), requiring merely $\mathtt{\lceil log_2(n) \rceil}$ qubits~\cite{10.5555/3309066}. Here, we describe the methodology to adapt the SWAP test quantum circuit~\citep {PhysRevLett.87.167902, 10.5555/3309066, kopczyk2018quantum} for $C_{\alpha}$ atoms distance matrix calculations.  

\subsubsection{Problem Size and Qubit Mapping}\label{ss:qmapping}

\textbf{Input matrix size}: While on classic architectures the amount of atoms and segments we can process is dependent on the amount of \texttt{RAM} available in the system, in the quantum machines we are limited by the amount of qubits of the machine.  As per the limitation of the target machines, we limit our input size, i.e. the length of amino acid segments used to generate the distance or bipartite matrix. For a chosen segment of length $k$ (consisting of k atoms), the Euclidean distance matrix $D$ (within the same segment) is a $k \times k$ symmetric matrix with diagonal entries as zeros. The constructed bipartite distance matrices $B_{\text{IJ}}$ (between two separate segments) is a $2\text{k} \times 2\text{k}$ dimensional matrix with a $\text{k} \times \text{k}$ block matrices $E_{\text{IJ}}$ (cf. \hyperref[matrix-size]{Appendix E}). 

Exploiting symmetries of the input CVs leads to a significant dimensional reduction on the matrix sizes, making it feasible to encode smaller input sizes onto quantum devices. In our case, it suffices to calculate only $\frac{k (k+1)}{2}$ unique entries of $D$ and $k^2$ entries for the block matrix of $E_{\text{IJ}}$ of the $2k \times 2k$ sized  $B_{\text{IJ}}$ (since the other block is just the tranpose of the matrix~cf. Appendix E) respectively. Due to such dimensional reduction properties intrinsic to our MD system, it becomes viable to cut-down the input system size for our Target Task I and Target Task II on the NISQ hardware.
\medskip

The \textit{quantum state preparation} for the $C_{\alpha}$ atom coordinates done via pseudo code \ref{alg:data-driven} is depicted in the \emph{Data Encoding} block in Figure \ref{fig:cswap}.

\begin{figure}[H]
    \centering
    \includegraphics [width=0.90\columnwidth]{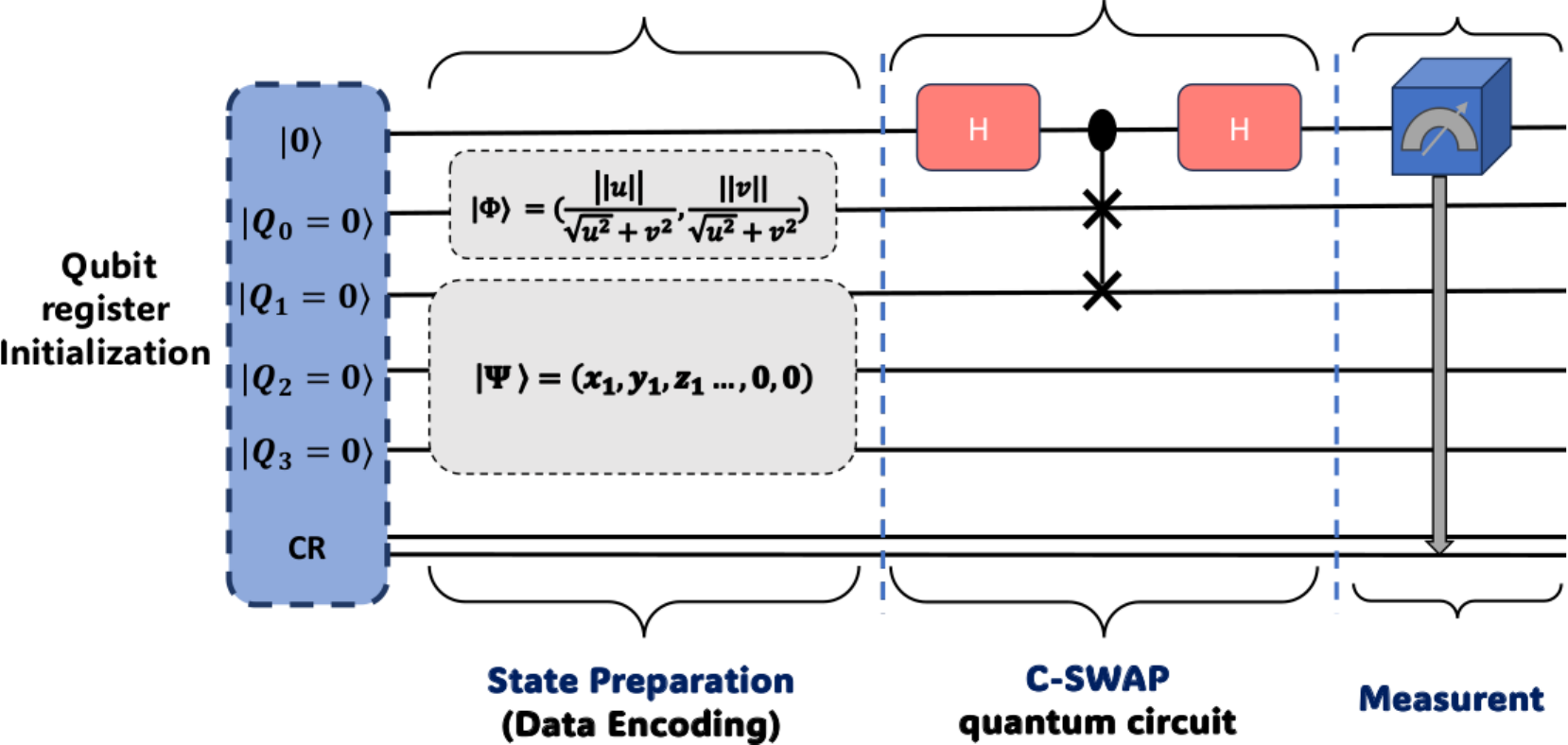}
    \caption{\texttt{C-SWAP} quantum subroutine implementation for distance matrix generation.}
    \label{fig:cswap}
\end{figure}

Let the position vectors (classical data) corresponding to two different atoms be  denoted as $\Vec{u},\Vec{v}$. Then the qubits can initialized with the state amplitudes (coefficients), in such a way that:
\begin{subequations} \label{encoding}
\begin{equation}
\ket{\phi} = \frac{1}{\sqrt{W}}(||\vec{u}||\ket{0} - ||\vec{v}||\ket{1} ),
\end{equation}
\begin{equation}
\ket{\psi} = \frac{1}{\sqrt{2}}(\ket{u, 0} + \ket{v, 1}),
\end{equation}
\end{subequations}
with $||\vec{u}||$, $||\vec{v}||$ being the Euclidean norm of the coordinates  and $W = ||\vec{u}||^2 + ||\vec{v}||^2$. The corresponding amplitude-encoded vectors read,
\begin{subequations}\label{encodingII}
\begin{equation}
\ket{u} = \sum_{i = 0}^{N-1} \frac{u_i}{||\vec{u}||}\ket{i},    
\end{equation}
\begin{equation}
\ket{v} = \sum_{i = 0}^{N-1} \frac{v_i}{||\vec{v}||}\ket{i} 
\end{equation}
\end{subequations}
Since each qubit has two possible states, the number of coordinates of the classical vector data must necessarily be $2^n$ \cite{10.5555/3309066}.
The real-time coordinates of $C_{\alpha}$ atoms whose positions vectors evolves as a function of time $t$,  $\{x_i (t), y_i (t), z_i (t)\}_{1 \leq i \leq n} \in \mathbb{R}^3$. These three-dimensional position vector $\Vec{u}$ must be padded with a $0$ as the fourth coordinate. This leads to a vector of the form $\big(x(t),y(t),z(t),0\big)$ with $2^2 = 4$ coordinates. This allows a suitable encoding scheme onto a $2$-qubit quantum register. 
\medskip 

In Figure \ref{fig:cswap}, the $3$-qubit quantum register $\ket{\psi}=$ $\ket{Q_1, Q_2, Q_3}$, initialized as $\ket{0, 0, 0}$ is then encoded with the concatenated atom-pair coordinates values\footnote{For the sake of brevity, we drop the time-dependence of the atom coordinates.} $(x_1, y_1, z_1, x_2, y_2, z_2, 0, 0) \in \mathbb{R}^8$. Here, it is clear that $\mathtt{0}$ padding is required for embedding, since three qubits are only capable of storing a vector of dimension $2^3 = 8$. 

The is followed by an application of the \texttt{HADAMARD} gate \texttt{H} (cf. Appendix B) on the ancillary qubit. Implementing controlled swap operation, is performed using the three-qubit \texttt{FREDKIN} gate (cf. Appendix B)  on the other two registers. The ancillary qubit works like a control bit. The total state of the system after these two gate operations is,
\begin{equation}
\ket{\Psi_B} = \frac{1}{\sqrt{2}}\big(\ket{0,\psi,\phi} + \ket{1,\phi,\psi}\big).
\end{equation}
The application of another Hadamard gate on the ancillary qubit $\ket{0}$ yields
\begin{equation*}
\frac{1}{2}\ket{0}(\ket{\phi,\psi}+\ket{\psi,\phi}) + \frac{1}{2} \ket{1}(\ket{\phi,\psi}+\ket{\psi, \phi}).
\end{equation*}
On applying the Hadamard gate, the  probability of measuring state $'0'$, i.e., of control qubit yields,
\begin{equation*}
    Pr(0) = \frac{1}{2} + \frac{1}{2}|\bra{\phi}\ket{\psi}|^2.
\end{equation*}
Euclidean distances $d(i, j)$ between $C_{\alpha}$ atoms can be obtained using Eqs.( \ref{encoding},~\ref{encodingII}).
\begin{equation}
   d(\vec{u}, \vec{v})^2 = 2 W |\bra{\phi} \ket{\psi}|^2 = 4 W (Pr(0) - 0.5).
\end{equation}
\medskip 

Our simulations comprised of a similar strategy as put-forth in~\cite{PhysRevLett.87.167902}. Here, we perform a single execution for each $C_{\alpha}$ atom pair using the SWAP test subroutine.  Hence, a total number of $n$ repeated circuit executions were required for calculating the distances between $n$ atom pairs using the quantum architecture.

\section{The Variational Quantum Eigensolver Machinery}
\renewcommand{\theequation}{F.\arabic{equation}}
\setcounter{equation}{0}
\subsection{The mathematics of VQE} 
The theoretical groundwork for VQE starts with the variational Rayleigh-Ritz functional. Given a Hamiltonian (Hermitian operator) $\hat{H}$ and a intial trial wavefunction with respect to some vector-valued parameter $\boldsymbol{\vartheta}$ is $\ket{\Psi (\boldsymbol{\vartheta)}}$ (ansatz wavefunction). The Rayleigh-Ritz variational principle~\cite{Ritz1909} sets an optimized upper bound for the ground state energy $E_0$ (lowest possible expectation/average value) associated with the Hamiltonian \footnote{The expectation value of a matrix $\hat{O}$ with respect to a vector $\ket{\phi}$ is defined as  $\frac{\bra{\phi}\hat{O}\ket{\phi}}{\bra{\phi}\ket{\phi}}$.},  $E_0$, i.e., 
\begin{equation}\label{vqe}
    E_0 := \langle \hat{H} \rangle_{\boldsymbol{\vartheta}} \leq \frac{\bra{\Psi(\boldsymbol{\vartheta})}\hat{H (\boldsymbol{\vartheta})}\ket{\Psi (\boldsymbol{\vartheta})}}{\bra{\Psi (\boldsymbol{\vartheta})}\ket{\Psi (\boldsymbol{\vartheta})}}.
\end{equation}
The VQE machinery finds a paramterization of the wavefunction $\ket{\Psi}$, such that the \texttt{expectation value} of the Hermitian operator $\hat{H}$ is minimized and approaches closer to the lowest eigenvalue $E_0$ after successive iterative optimization steps~\cite{Tilly2021}.   
\medskip 

A PQC consisting of an initialized qubit register and a set of unitary quantum gates, can only perform a series of unitary transformations and measurements. Inorder to execute such a minimization (optimization) task as described in Eq.(\ref{vqe}) using quantum circuits, the user must define a  so-called ansatz wavefunction (trial eigenvector) $\ket{\Psi(\boldsymbol{\vartheta})}$. An initial  generic parametrized unitary quantum gate  $U(\boldsymbol{\vartheta})$ applied onto an initialized qubit register state, say $U(\boldsymbol{\vartheta}) \ket{0}^{\otimes N} = \ket{\Psi(\boldsymbol{\vartheta})}$ ($\forall$ $\boldsymbol{\vartheta} \in (-\pi, \pi]$) generates the ansatz wavefunction. 
\medskip 

The Hamiltonian $\hat{H}$ (Hermitian Matrices in general) can be encoded onto the Pauli operators multiplied by weights (linear combination of elements in the Pauli group), i.e., 

\begin{equation}
    \hat{H} = \sum_{\alpha} w_a \hat{P}_a, \ \forall  \hat{P}_a \in \mathcal{P}_n.
\end{equation}
Here, $w_a$ are the set of weights (coefficients) and $\hat{P}_a$ are Pauli strings in $\mathcal{P}_n$respectively~\cite{Tilly2021}. In the Pauli decomposed version, Eq.(\ref{vqe:cost})
 Thus the VQE optimization problem, designed using the quantum circuit reads,
\begin{align}\label{vqe:cost}
    E_{VQE} = \min_{\theta} \bra{\boldsymbol{0}} U^{\dagger} (\boldsymbol{\theta}) \hat{H} U(\boldsymbol{\theta}) \ket{\boldsymbol{0}} = \min_{\boldsymbol{\theta}} \sum_{a}^{\mathcal{P}} w_a \bra{\boldsymbol{0}} U^{\dagger} (\boldsymbol{\theta}) \hat{P_a} U(\boldsymbol{\theta}) \ket{\boldsymbol{0}}.
\end{align}
The iterative optimization of Eq.(\ref{vqe:cost}) is similar to the one that one encounters in machine learning , thus, is also known as the cost (loss) function in Hybrid systems.



Thus, the quantum expectation values need to be executed on a quantum device. By contrast, operations like summation of the expectation values in Eq.(\ref{vqe:cost}) and the iterative optimization, for e.g., gradient-descent (parameter update à la machine learning) of each of the terms in  $E_{VQE} = \min_{\boldsymbol{\theta}} \sum_{a} w_a E_{p_a}$, is carried out using classical optimization algorithms. This clearly depicts the workload sharing pipeline between classical and quantum devices in hybrid frameworks.

\end{document}